\documentclass[runningheads]{llncs}
\usepackage{booktabs} 

\usepackage{tikz}
\usetikzlibrary{positioning,chains,fit,shapes,calc}

\newcommand{\mb}{\ensuremath{\operatorname{MB-OSM}}\xspace}
\newcommand{\rank}{\ensuremath{\operatorname{Ranking}}\xspace}
\newcommand{\gre}{\ensuremath{\operatorname{Greedy}}\xspace}
\newcommand{\att}{\ensuremath{\operatorname{ATT}}\xspace}
\newcommand{\Ber}{\ensuremath{\mathrm{Ber}}\xspace}
\newcommand{\Var}{\ensuremath{\mathsf{Var}}\xspace}

\newcommand{\ra}{\mathsf{Ratio}}
\newcommand{\tbf}{\textbf}
\newcommand{\sfg}{\mathsf{g}}
\newcommand{\be}{\bar{e}}
\newcommand{\bt}{\bar{t}}

\usepackage{macros_pan}
\usepackage[sectionbib,numbers,compress]{natbib}
\begin{document}

\title{Exploring the Tradeoff between Competitive Ratio and Variance in  Online-Matching Markets\thanks{This paper was accepted to the 18th Conference on Web and Internet Economics (WINE), 2022. PX was partially supported by NSF CRII Award IIS-1948157. The author would like to thank the anonymous reviewers for their valuable comments.}}

\author{Pan Xu}
\institute{Department of Computer Science, New Jersey Institute of Technology\\
\email{pxu@njit.edu}}

\maketitle

\begin{abstract}
In this paper, we propose an online-matching-based model to study the assignment problems arising in a wide range of online-matching markets, including online recommendations, ride-hailing platforms, and crowdsourcing markets. It features that each assignment can request a random set of resources and yield a random utility, and the two (cost and utility) can be arbitrarily correlated with each other. We present two linear-programming-based parameterized policies to study the tradeoff between the \emph{competitive ratio} (CR) on the total utilities and the \emph{variance} on the total number of matches (unweighted version). The first one (SAMP) is simply to sample an edge according to the distribution extracted from the clairvoyant optimal, while the second (ATT) features a time-adaptive attenuation framework that leads to an improvement over the state-of-the-art competitive-ratio result. We also consider the problem under a large-budget assumption and show that SAMP achieves asymptotically optimal performance in terms of competitive ratio.
\end{abstract}


\section{Introduction} Online-matching models have been used to study a wide range of real-world matching markets, including ride-hailing platforms,  online recommendations, and crowdsourcing markets.  One of the central problems there is to design an online-matching policy such that the expected total utility is maximized over all matches. Most of the current models assume that each match requests one single deterministic resource and what is more, the profit gained through the match is typically supposed to be independent of the amount of resource requested. These assumptions make current models hard to capture many real applications perfectly. Consider the following  three motivating examples.

\xhdr{Bundle recommendations}. Consider online bundle recommendations~\citep{zhu2014bundle}. We have a ground set $\cU$ of all offline items to sell. Upon the arrival of an online buyer, say $j$, we need to select a bundle $S \subseteq \cU$ of offline items to offer to $j$, and then the user $j$ will accept and reject $S$ with certain respective probabilities, say $p$ and $1-p$. Assume the acceptance of bundle $S$ will yield some profit, say $w_{j,S}$, to the platform (\eg Amazon) and each type of offline item has a limited number of copies in stock. In this case, we observe that after ``matching'' $S$ with $j$: With probability $p$, we will deplete a copy of all items in $S$ and get a profit $w_{j,S}$, and with probability $1-p$, it will incur no cost and no profit.

\xhdr{Display advertising}. Consider display advertising business~\citep{zhu2017optimized,abbassi2015optimizing}. We have a ground set $\cU$ of all offline impressions (or ads). Upon the arrival of an online user of type $j$, one ads platform (\eg Google) will display to her a set of ads, say $S \subseteq \cU$. Then, the user $j$ will select a subset $S' \subseteq S$ to click, which occurs with some probability $p_{j,S'}$, and this yields profit $w_{j,S'}$ to the ads platform as a result. Assume each ad  has a displaying capacity  due to the budget of the advertiser. In this context, matching $j$ with $S$ will lead to a random consumption of budgets of ads in $S$ and a random profit, and the two (consumption and profit) are positively correlated with each other.

\xhdr{Task assignment in crowdsourcing markets}. Consider task assignment problem in crowdsourcing human-resource markets~\citep{ho2012online}, in which we crowdsource arriving workers to complete as many tasks as possible.
We have a ground set $\cU$ of offline tasks. Upon the arrival of an online worker of type $j$, the platform (\eg Amazon Mechanical Turk) assigns her a set $S \subseteq \cU$ of relevant tasks, and then the worker will select a subset $S^\prime \subseteq S$ of tasks to work according to her preferences, which occurs with a certain probability $p_{j,S'}$. Assume every task has a matching capacity reflecting the limited stock. In this context, after matching $j$ and $S$, with probability $p_{j,S'}$, we will exhaust one copy of each task in $S'$ and gain a utility of $|S'|$ (the number of tasks completed).

In all the three examples, each assignment may incur a random set of resources and yield a stochastic profit, and the cost and profit can be highly correlated with each other. In this paper, we propose a unifying model, called \emph{Multi-Budgeted Online Stochastic Matching} (\mb),  to study these applications.

\tbf{Multi-Budgeted Online Stochastic Matching (\mb)}. Suppose we have a bipartite graph $G=(I,J,E)$, where $I$ and $J$ denote the sets of offline and online agents, respectively.  We have a set of $K$ resources, denoted by $[K] \doteq \{1,2,\ldots, K\}$, and each resource $k \in [K]$ has a budget $B_k \in \mathbb{Z}^{+}$. Each edge $e \in E$ is associated with a random utility $W_e \ge 0$ and a random vector-valued cost, $\cA_e=(A_{e,k})_{k\in[K]}$, which takes values over $\{0,1\}^K$.  We define the support of $\cA_e$ as $\cS_e=\{k \in [K]: \E[A_{e,k}]>0\}$. In our context, $\cA_e$ with $e=(i,j)$ captures the random set of resources requested by matching agents $i$ and $j$, and $\cS_e$ represents the set of all possible resources potentially consumed by $e$.  Note that offline vertices in $I$ are static, while online vertices in $J$ arrive dynamically. Specifically, we consider a given time horizon $T$. For each time (or round) $t \in [T] \doteq\{1,2,\ldots, T\}$, one single online vertex $\hat{j}$ will be sampled (called $\hat{j}$ arrives) following a known distribution $\{p_{j}\}$ \emph{with replacement} such that $\Pr[\hat{j}=j]=p_j$ for each $j\in J$ with $\sum_{j \in J} p_{j}=1$. Note that the sampling process is independent and identical across the $T$ rounds. For each $j$, let $r_j=T \cdot p_j$, which is called the \emph{arrival rate} of $j$ with $\sum_{j \in J} r_j=T$. Our arrival assumption is commonly referred to as the \textit{known identical independent distributions} (KIID).\footnote{KIID is mainly inspired from the fact that we can often learn the arrival distribution from historical logs~\citep{Yao2018deep,DBLP:conf/kdd/LiFWSYL18,DBLP:conf/kdd/WangFY18}. It is widely adopted to study practical online-matching markets~\cite{xu-aaai-19,dickerson2018assigning,fata2019multi}, and it is also a common setting in theoretical online-matching models~\citep{feldman2009online,huang2021online,haeupler2011online,manshadi2012online,jaillet2013online,brubach2016new}.} Upon the arrival of an online vertex $j$, an immediate and irrevocable decision is required before observing the next arrival of online vertex: either reject $j$ or match it with an offline neighbor $i$ with $e=(i,j)\in E$. In the latter case, instant cost and utility realizations will be observed, \eg $\cA_e=(a_{e,k})$ and $W_e=w_e$, and the budget of each resource $k$ will get reduced by $a_{e,k}$ and we will gain a utility of $w_e$ as a result.

\tbf{Remarks on the model of \mb}. (1) The following information is all known as part of the input and is accessible to the algorithm: $G=(I,J,E)$, $\{B_k| k \in [K]\}$, distributions of $\{\cA_e, W_e| e\in E\}$, $\{p_j, r_j| j\in J\}$ and $T$. \emph{Additionally, we assume $T \gg 1$, and some of our results are obtained by taking $T\rightarrow\infty$}, a common practice in studying competitive ratio for theoretical online-matching models under KIID~\citep{haeupler2011online,manshadi2012online,jaillet2013online}. (2) Cost distributions among all edges ($\{\cA_e|e\in E\}$) are independent; the same for all utility distributions of $\{W_e|e\in E\}$. For each given edge $e$, its cost $\cA_e$ and utility $W_e$ can be arbitrarily correlated, however. (3) \emph{Throughout this paper, we assume each edge will incur at most $\Del$ different possible resources, \ie $|\cS_e| \le \Del$ for every $e\in E$, where $\Del$ is called sparsity}. The sparsity typically takes a small constant value in practice,\footnote{This can be seen from the fact that sparsity captures the tolerance on the bundle size of online buyers, the patience on the number of ads displayed to on online user simultaneously, and working capacity among online workers in the three aforementioned applications, respectively.}   though the total number $K$ of all resources can be huge. Fortunately, as shown later, the performance of algorithms proposed here will rely only on the sparsity $\Del$, regardless of $K$. (4) Safe policies. We say a policy (or algorithm) \ALG is \emph{safe} if $\ALG$ can make an assignment $e$ only when it will not violate any budget constraint for any possible realizations of $\cA_e$ (\ie every resource in $\cS_e$ should have at least one unit budget remaining  then). Safe policies are required in most real-world applications. For example, we can offer a bundle to a buyer only when all items in the bundle have at least one copy in stock. \emph{For this reason, we consider safe policies only throughout this paper}.

\tbf{Remarks on sources of randomness in \mb}. Consider a given (randomized) policy $\ALG$. Let $\cM$ be a  set of (random) assignments. The  expected total amount of utilities obtained by $\ALG$ is defined as $\E[\ALG]=\E[\sum_{e\in \cM} W_e]$, where the expectation is taken over the following four sources of randomness: \tbf{(R1)} dynamic arrivals of online agents over the $T$ rounds;  \tbf{(R2)} randomness possibly used by the policy $\ALG$;  \tbf{(R3)} randomess in cost realizations of $\{\cA_e\}$;  and \tbf{(R4)} randomess in utility realizations of $\{W_e\}$.
 Our goal is to design a policy that achieves  as large utilities as possible while with a variance as small as possible.
 
\subsection{Preliminaries}
Throughout this paper, we set $[n]:=\{1,2,\ldots,n\}$ for any integer $n$.

\xhdr{Competitive ratio}. Competitive ratio (CR) is a commonly used metric to evaluate the performance of online algorithms. Consider maximization of the total utilities in \mb as studied here for example. Consider a given algorithm $\ALG$ and an instance $\cI$ of \mb.
Let $\ALG(\mathcal{I})$ and $\OPT(\cI)$ be the expected total utilities achieved by $\ALG$ and a clairvoyant optimal $\OPT$ on $\cI$, respectively. We say $\ALG$ achieves a competitive ratio of $\rho \in [0,1]$ if $\ALG(\cI) \ge \rho \OPT(\cI)$ for all possible instances $\cI$ of \mb. 

Here are a few similarities and differences between $\ALG$ and $\OPT$: (1) The expected performance of $\ALG$ and $\OPT$ are both taken over the four sources of randomness (\ie \tbf{R1-R4}); (2)  $\ALG$ is required to make an immediate matching decision upon every arrival of an online vertex before the next one, while  \OPT enjoys the privilege of accessing the full arrival sequence of online vertices before any decisions; (3) Neither \ALG nor \OPT has access to realizations of the cost or utility of an edge until the edge has been added; (4) Both \ALG and \OPT should follow rules of safe policies, \ie an edge $e$ can be added only when every resource in $\cS_e$ has at least one remaining budget. The toy example below shows that the natural two heuristics, \gre and \rank, both achieve a competitive ratio of zero on \mb.\footnote{Note that \gre and \rank operate as follows:  \gre matches each arriving $j \in J$ with a neighbor $i$ such that $e=(i,j)$ has the largest expectation of utility among all safe assignments (\ie no budget violation will be caused); \rank  first chooses a random order $\pi$ over $I$ and then matches each arriving $j$ with a neighbor $i$ such that $i$ has the lowest order in $\pi$ among all safe choices.} 

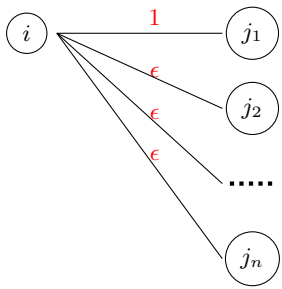
\begin{figure}[th!]
\begin{minipage}{.5\linewidth}
 \begin{tikzpicture}
 \draw (0,0) node[minimum size=0.2mm,draw,circle] {$i$};
  \draw (3,0) node[minimum size=1mm,draw,circle] {$j_1$};
    \draw (3,-1) node[minimum size=1mm,draw,circle] {$j_2$};
            \draw (3,-3) node[minimum size=1mm,draw,circle] {$j_n$};
             \node [red,above] at (1.7,0) {$1$};
                \node [red,above] at (1.7,-0.7) {$\ep$};
                \node [red,above] at (1.7,-1.25) {$\ep$};
                     \node [red,above] at (1.7,-1.8) {$\ep$};
\draw[-] (0.4,0)--(2.6,0);
\draw[-] (0.4,0)--(2.6,-1);
\draw[-] (0.4,0)--(2.6,-3);
\draw[-] (0.4,0)--(2.6,-2);
\draw [dotted, ultra thick](2.7,-2) -- (3.3,-2);
\end{tikzpicture}
\end{minipage}\hfill
\begin{minipage}{.5\linewidth}
\begin{align*}
&I=\{i\}, J=\{j_1, \ldots, j_n\}, E=\{(i,j_\ell)| \ell \in [n]\}; \\
&K=1, B=1; \\
& \cA_e=1 \mbox{ with probability $1$}, \forall e \in E;\\
&W_{(i,j_1)}=1, \mbox{ with probability $1$};\\
&W_{(i,j_\ell)}=\ep, \mbox{ with probability $1$}, \forall 1<\ell \le n;\\
&T=n,p_j=1/n, r_j=1, \forall j \in [n];\\
&\gre=\rank\le \ep+1/n;\\
&\OPT \ge 1-1/\sfe.
\end{align*}
\end{minipage}
\caption{A toy example on which  \gre and \rank both achieve a competitive ratio of zero.} \label{fig:exam-1}
\end{figure}

\begin{example}[\gre and \rank both achieve a competitive ratio of zero] Consider such a toy example as shown in Figure~\ref{fig:exam-1}. We have a star graph with $I=\{i\}$ and $J=\{j_1,\ldots,j_n\}$, $T=n$,  and $r_j=1$ for every $j \in J$. Thus, during each round $t \in [T]$, one single online agent $\hat{j}$ will be sampled uniformly at random with replacement such that $\Pr[\hat{j}=j_\ell]=1/n$ for every $\ell \in [n]$. We have one single resource with a unit budget, and each edge will cost one unit resource. The edge $(i, j_1)$ has a deterministic weight of one, while the rest have a deterministic weight of $\ep>0$. Our example captures a simple instance of the classical online matching under KIID, where the offline vertex $i$ can be viewed as the single resource with a unit matching capacity.

 We can verify that (1) $\gre$ and $\rank$ reduce to the same and both achieve an expected utility of $1/n \cdot 1+(1-1/n) \cdot \ep \le 1/n+\ep$; (2) \OPT (a clairvoyant optimal) achieves an expected utility of $(1-1/\sfe) \cdot 1+ 1/\sfe \cdot \ep$, where $\OPT$ will assign $i$ to $j_1$ if $j_1$ arrives at least once (that happens with probability $1-1/\sfe$). By definition,  \gre and $\rank$ both achieve a competitive ratio no more than $(1/n+\ep)/(1-1/\sfe) \rightarrow 0$ when $n \rightarrow \infty$ and $\ep \rightarrow 0$. \hfill $\blacksquare$
\end{example}

\tbf{Variance analysis}. In this paper, we pioneer variance analysis for online algorithms in the context of online stochastic matching under known distributions. For most classical optimization problems formulated in an offline setting where input information is fully accessible  (\eg finding a maximum weighted matching in a general graph), we are allowed to run an algorithm multiple times on a given input instance. In this context, suppose we design a randomized algorithm $\ALG$ and show that its \emph{expected} performance ($\E[\ALG]$) is good enough. By applying  de-randomization techniques like conditional expectations~\cite{raghavan1988probabilistic}, we can get a \emph{deterministic} version that can performance as good as $\E[\ALG]$. Note that for online optimization problems like \mb, we can run an online algorithm \emph{only once} on a given input instance and thus, de-randomization techniques fail to work here. This highlights the importance of variance analysis in online-algorithm design, in addition to the popular competitive-ratio (CR) analysis. Observe that the CR metric reflects only the gap between an online algorithm (\ALG) and a clairvoyant optimal (\OPT) in terms of their  \emph{expected} performance: it shows no any guarantee on the variance or robustness of \ALG. 

In this paper, we focus on analyzing the variance on the total (random) number of matches instead of the total utilities achieved by any algorithm. Note that the randomness in utility realizations (\tbf{R4}) can contribute to an unbounded variance in the total utilities even for simple deterministic algorithms. Consider such a toy example as follows: There is only one single edge $e=(i,j)$ in the graph and one single round $T=1$ where $j$ will arrive with probability one; We have one single resource with a unit budget; $A_e=1$ deterministically while $W_e=1/\ep$ with probability $\ep$ and $W_e=0$ otherwise. We can verify that the total utilities achieved by any policy adding edge $e$ will have a variance equal to $1/\ep-1$, which can be arbitrarily large. For this reason, we focus on analyzing the variance on the total number of matches, which is due to randomness sources as outlined in \tbf{R1}, \tbf{R2}, and \tbf{R3} only (see \tbf{Remarks on sources of randomness} in \mb).

\subsection{Main Contributions} 
Our contributions are summarized as follows. First, we present a canonical LP (Section~\ref{sec:lp}) as the benchmark, whose optimal value proves a valid upper bound on the total expected utilities achieved by a clairvoyant optimal. Second, we design two LP-based parameterized policies to study the tradeoff between the competitive ratio (CR) on the total utilities and the variance on the total number of matches. The first one  (\samp) is simply to sample an edge according to the distribution extracted from a clairvoyant optimal, while the second (\att) features a time-adaptive attenuation framework. In the last, we study a special case under the large-budget assumption and show that the first algorithm \samp can achieve an asymptotically optimal CR that approaches one when budgets go infinity. Here are the details.

\begin{theorem}\label{thm:main-sa}[Section~\ref{sec:samp}]
There exists a parameterized LP-based sampling algorithm $\samp(\alp)$  with $\alp \in [0,1]$ such that (i) it achieves a competitive ratio (CR) \emph{equal to} $(1-\sfe^{-\alp \Del})/\Del$ on the total utilities with respect to the benchmark LP~\eqref{obj-1}; and (ii) it achieves a variance at most $(\alp T)^2 \cdot \sfg\big(\min(\Del \alp, \eta)\big)+O(T)$ on the total number of matches, where $\sfg(x):=\big(1 - \sfe^{-2 x} - 2x \sfe^{-x}\big)/x^2$, and $\eta\sim1.126$ is the unique maximizer of $\sfg(x)$ when $x \in [0,\infty]$. Both the CR and variance analyses are tight.\footnote{Tightness on the analysis means we can identify an instance on which the CR (or variance) achieved by $\samp$  (or \att) matches the claimed bound.}
\end{theorem}

\begin{theorem}\label{thm:main-att}[Section~\ref{sec:att}]
There exists a parameterized sampling algorithm $\att(\alp)$ with $\alp \in [0,1]$ such that (i) it achieves a competitive ratio (CR) \emph{equal to} $(1-\sfe^{-\alp \Del})/\Del$ on the total utilities with respect to the benchmark LP~\eqref{obj-1}; and (ii) it achieves a variance at most $(\alp T)^2 \cdot \sfg(\alp \Del)+O(T)$ on the total number of matches, where $\sfg(x):=\big(1 - \sfe^{-2 x} - 2x \sfe^{-x}\big)/x^2$. Both CR and variance analyses are tight.
\end{theorem}
\xhdr{Remarks on $\samp(\alp)$ and $\att(\alp)$}. (i) Both $\samp(\alp)$ and $\att(\alp)$ achieve a competitive ratio (CR)  \emph{equal} to $(1-\sfe^{-\alp \Del})/\Del$ on the total expected utilities. Note that $(1-\sfe^{-\alp \Del})/\Del $ is an increasing function of $\alp \in [0,1]$ for any given integer $\Del$. Meanwhile, we can verify that the upper bounds of variance for $\samp(\alp)$ and $\att(\alp)$ are both increasing functions of $\alp \in [0,1]$ for any given integer $\Del$. This suggests that when $\alp$ takes a larger value, both algorithms will achieve a higher CR on the total utilities at a price of a higher variance on the total number of matches and vice versa. (ii) Recall that CR is defined based on the worst-case on which the ratio of the performance of an algorithm to that of a clairvoyant optimal gets minimized. Though $\samp$ and $\att$ achieve the same worst-case CR, $\att$ shows more robust than \samp in the way that $\att(\alp)$ will achieve a CR equal to $(1-\sfe^{-\alp \Del})/\Del$ on \emph{every} input instance, while $\samp(\alp)$ will achieve the worst-case CR on a very specialized instance as identified in Example~\ref{exam:cr-ws}. This partially explains why the upper bound of variance of $\samp$ is slightly larger than that of $\att$. (iii) Though $\att$ is conceptually more complicated than $\samp$, both the CR and variance analyses of $\att$ turn out much simpler than those of $\samp$.

We complement the above lower bounds of CR by showing some upper bound due to the benchmark LP (conditional hardness result).
\begin{theorem}\label{thm:hard}[Section~\ref{sec:hardness}]
No algorithm can achieve a competitive ratio better than  $\big(1-\sfe^{-(\Del-1+1/\Del)}\big)/(\Del-1+1/\Del)$ with respect to the benchmark \LP~\eqref{obj-1} with $\Del-1$ being a prime.
\end{theorem}
\xhdr{Remarks on CR for online-matching models related to \mb}. (i) When $\Del=1$, the upper bound in Theorem~\ref{thm:hard} matches the best lower bound of $1-1/\sfe$ (which arrives at $\alp=1$) in Theorems~\ref{thm:main-sa} and~\ref{thm:main-att}. This suggests the tightness of lower bounds on CR with respect to the current benchmark \LP~\eqref{obj-1}.\footnote{Tightness here means that the lower bound of CR is the best we can get based on the current benchmark LP.} (ii) \citet{aamas-19} considered a special case of \mb when each edge is associated with a \emph{deterministic} vector-valued cost with sparsity $\Del$ but under known adversarial distributions, which allow the arrival distributions to change over time. They gave an upper bound (hardness result) of $1/(\Del-1+1/\Del)$ and a lower bound of $1/(\Del+1)$, respectively.  Note that both the lower and upper bounds are improved here: the best lower bound of $\samp$ and $\att$ (when $\alp=1$) satisfies $(1-\sfe^{-\Del})/\Del \ge 1/(\Del+1)$; see Figure~\ref{fig:comp}. (iii) \citet{brubach2016new} considered a variant of online stochastic matching, which can be cast as a \emph{strict} special case of \mb in the way that each edge is associated with a Bernoulli random vector-valued cost with $\Del=1$. From Theorems~\ref{thm:main-sa} and~\ref{thm:main-att}, both $\samp(\alp)$ and $\att(\alp)$ achieve a CR of $1-1/\sfe$ with $\alp=1$ when $\Del=1$, which matches that of~\citep{brubach2016new}. (iv)~\citet{kess13} considered online $\Del$-hypergraph matching that can be cast as a special case of our model where each edge takes a deterministic cost vector with sparsity $\Del$. They considered the random arrival order and gave a CR of $1/(\sfe \cdot \Del)$, which is much worse than the best CR as stated in Theorems~\ref{thm:main-sa} and~\ref{thm:main-att} that is equal to $(1-\sfe^{-\Del})/\Del$ when $\alp=1$.

\begin{figure}[th!]
\begin{center}
 \begin{tikzpicture}[ultra thick]
\draw [->] (0,0) --(13.5,0);
 \draw (12,0) node[minimum size=0.2mm,draw,circle,fill=red] {};
  \draw (12,0.2) node[above] {$\frac{1}{\Del-1+1/\Del}$};
   \draw (12,-0.2) node[below] {UB in~\cite{aamas-19}};

    \draw (9,0) node[minimum size=0.2mm,draw,circle,fill=red] {};
  \draw (9,0.2) node[above] {$\frac{1-\sfe^{-(\Del-1+1/\Del)}}{\Del-1+1/\Del}$};
   \draw (9,-0.2) node[below] {UB in Theorem~\ref{thm:hard}};

       \draw (5,0) node[minimum size=0.2mm,draw,circle,fill=blue] {};
    \draw (5,0.2) node[above] {$\frac{1-\sfe^{-\Del}}{\Del}$};
   \draw (5,-0.2) node[below] {LB in Theorems~~\ref{thm:main-sa} and~\ref{thm:main-att} ($\alp=1$)};

    \draw (1.5,0) node[minimum size=0.2mm,draw,circle,fill=blue] {};
    \draw (1.5,0.2) node[above] {$\frac{1}{\Del+1}$};
   \draw (1.5,-0.2) node[below] {LB in~\cite{aamas-19}};
   
\end{tikzpicture}
\end{center}
\caption{Upper bounds (UB) and lower bounds (LB) on competitive ratio of \mb as shown in this paper and in~\cite{aamas-19}: Note that $(1-\sfe^{-\Del})/\Del \ge \big(1-1/(\Del+1)\big)/\Del=1/(\Del+1)$, where $\Del:=\max_{e \in E} |\cS_e|$.}
\label{fig:comp}
\end{figure}
In the last, we consider $\mb$ under the large-budget assumption. Let $B=\min_{k \in [K]} B_k$, which denotes the minimum budget over all resources. When all budgets are large with $B \gg 1$, we show that $\samp(\alp)$ with $\alp=1$ can achieve an asymptotically optimal CR with respect to the benchmark \LP~\eqref{obj-1}.
\begin{theorem}\label{thm:main-3}[Section~\ref{sec:large}]
There exists an LP-based sampling algorithm $\samp(1)$ such that 
(i) it achieves a competitive ratio (CR) of $1-\frac{1}{\sqrt{2\pi B}}\big(1+o(1)\big)$ when $\Del=1$ and $T \gg B \gg 1$; (ii) it achieves  a competitive ratio of $1- \kappa \cdot  \sqrt{\frac{\ln \Del}{B}}(1+o(1))$ with $\sqrt{2}<\kappa \le 2\sqrt{2}$ when $T \gg B \gg \ln \Del \gg 1$. The competitive ratios for cases (i) and (ii) are both asymptotically optimal  with respect to the benchmark \LP~\eqref{obj-1}. 
\end{theorem}
\xhdr{Remarks on results of Theorem~\ref{thm:main-3}}. (1) The term $o(1)$ in part (i) vanishes when $B \rightarrow \infty$, while that in part (ii) vanishes when $B=\omega(\ln \Del)$ and $\Del \rightarrow \infty$.  (2) As mentioned before,~\citet{brubach2016new} considered a strictly special case of \mb with $\Del=1$. When all budgets are large, they gave an algorithm achieving an asymptotical online ratio of $1-B^{-1/2+\ep}(1+o(1))$ for any given $\ep>0$. Our result in part (i) significantly improves that. (3) The constants of the leading terms as stated in parts (i) and (ii) are asymptotically
optimal with respect to the current benchmark \LP~\eqref{obj-1}, which are $\frac{1}{\sqrt{2\pi}}$ and $\kappa$, respectively. That is to say, \eg no policy can achieve a CR of $1-\frac{c}{\sqrt{2\pi B}}(1+o(1))$ with a constant $c<1$ when $B \gg 1$ and $\Del=1$ if compared against the optimal value of \LP~\eqref{obj-1} (which proves a valid upper bound on a clairvoyant optimal). (4) The result in part (i) (\ie $1-\frac{1}{\sqrt{2\pi B}}(1+o(1))$) appears in multiple contexts before, including Adwords and correlation gap~\cite{devanur2012asymptotically,alaei2012online,yan2011mechanism}. However, the analysis here is essentially different from there. Let $\Ber(B/T)$ denote a Bernoulli random variable with mean $B/T$. For Adwords and related applications~\cite{devanur2012asymptotically,alaei2012online,yan2011mechanism}, they all care about $\E[\min(X,B)]/B$ when $T \rightarrow \infty$, where $X$ is the sum of $T$ \iid $\Ber(B/T)$s. In contrast, we need to figure out $\E[\min(T',T)]/T$ when $T \rightarrow \infty$, where $T'$ is the number of copies of  \iid  $\Ber(B/T)$s needed such that the total sum is equal to $B$. Though we can draw a subtle connection between the two (see the proof in Section~\ref{sec:proof-part-a}), it is not straightforward to see the two are the same. 

\subsection{Main Techniques and Other Related Works}\label{sec:mt}
\tbf{Main techniques}. Overall, both the competitive-ratio (CR) and variance analyses in this paper feature a sequential identification process of the worst-scenario (WS) structure on which the exact lower bound (for CR) and the exact upper bound (for variance) are attained. For the algorithm-design part, our second parameterized policy ($\att$) features a time-adaptive attenuation framework, which differs from previous  time-oblivious attenuations widely used before~\citep{ma2014improvements,adamczyk2015improved,brubach2017,feng2019linear,xuAAAI18}. Generally speaking, in a time-oblivious attenuation framework, we set a uniform attenuation target, say a given constant $\gam=1/2$, such that a ``good'' event will happen with probability exactly \emph{equal to} $\gam$ for every online agent regardless of her arriving time. In contrast, we propose a time-adaptive attenuation framework, where we carefully craft an attenuation target function $\gam(t)$,  which is adaptive to the arriving time $t$ of an online agent. This partially leads to an improvement on CR for \mb over the previous work~\cite{aamas-19} that adopts a time-oblivious attenuation. For the variance-analysis part, we propose several balls-and-bins models to facilitate the analysis and exploit related \emph{negative-association} properties exclusively applied to balls-and-bins models; see~\cite{dubhashi1996balls,joag1983negative,shao2000comparison}.

\tbf{Other related works}. \mb falls under the family of \emph{online packing} problems. There have been a few studies investigating CR of \mb under the large-budget assumption but under the arrival setting of random arrival order (RAO), which is less restrictive than KIID as considered here; see~\citep{agrawal2014dynamic,agrawal2014fast,buchbinder2009online}. For RAO, a powerful algorithm-design paradigm is called the primal-and-dual approach; see the survey book~\citep{buchbinder2009design}. It is interesting to compare the setting studied here with those of~\cite{kess-stoc} and~\cite{jacm19}, both of which considered online resource allocation with multiple budget constraints under RAO. However, they assumed that both the cost vector and utility are deterministic for each edge. In contrast, we assume each edge can have a random vector-valued cost and a random utility, and the two can be arbitrarily correlated with each other. Another difference is that they considered fractional cost in the way that each $\cA_e \in [0,1]^K$ (deterministic), while we assume that each $\cA_e \in \{0,1\}^K$ (random) here.~\citet{jacm19} gave a CR of $1-\omega(\sqrt{\ln \Del/B})$, which was improved to $1-O(\sqrt{\ln \Del/B})$ by \citet{kess-stoc}. Note that the work of \cite{kess-stoc} has not identified any proper constant included inside the term $O(\sqrt{\ln \Del/B})$, whereas it is one of the main focuses in this paper. The offline version of \mb captures the stochastic $\Del$-set packing problem as a special case, which was introduced by~\citet{bansal2012lp}. They gave a $2\Del$-approximation algorithm for the stochastic $\Del$-set packing problem, which was improved to $\Del+o(\Del)$ by~\citet{TALG-20} later.~\citet{baveja2018improved} considered $\Del$-uniform stochastic hypergraph matching, which can be viewed as a special case of the stochastic $\Del$-set packing problem. They gave two approximation algorithms that achieve a ratio of $\Del+1/2$ and a ratio of $\Del+\ep$ for any given $\ep>0$, respectively. 

There is a large body of research works that have studied budgeted online resource allocation in an \emph{online learning} setting, where distributions of utility and/or cost associated with assignments are unknown. In that context, a common practice is to formulate the problem as one of the renowned Multi-armed bandit variants~\cite{slivkins2019introduction} and then conduct regret analysis, showing the expected total regret (defined as the gap in the total utility achieved by a given policy and a prophet optimal) is upper bounded by a certain function of the total time horizon~\cite{wu2015algorithms,balseiro2019learning,balseiro2022best,golrezaei2021bidding}. A few recent works investigate the potential tradeoff between variance and regret in online learning; see, \eg ~\cite{van2022regret,vakili2019decision}. In particular,~\citet{vakili2019decision} introduced and analyzed the performance of several risk-averse policies in both bandit and full information settings under the metric of mean-variance~\cite{steinbach2001markowitz}.  


\xhdr{Glossary of notations}. We offer a glossary of notations used throughout this paper; see Table~\ref{table:notation}.

\begin{table}[ht!]
\caption{A glossary of notations used throughout this paper.}
\label{table:notation}
\begin{tabular}{ll } 
 \hline
  $[n]$ & Set of integers $\{1,2,\ldots,n\}$ for any generic integer $n$.  \\
$G=(I,J,E)$ & Input compatibility graph where $I$ and $J$ are sets of offline and online vertices.  \\
 $K$ & Toal number of resources. \\
  $B_k$ & Budget on the resource $k \in [K]$.\\ 
 $\cA_e=(A_{e,k})$ & Random vector-valued cost on edge $e$ with $\E[A_{e,k}]=a_{e,k}$.\\
 $\cS_e$ & Support of the cost of edge $e$, \ie $\cS_e=\{k: \E[A_{e,k}]>0\}$. \\ 
$\Del$ & Sparsity defined as the largest size of edge cost support, \ie $\Del=\max_{e\in E} |\cS_e|$. \\ 
$T$ & Total number of online rounds (times). \\
$p_j$ & Probability that online vertex $j$ arrives during each round.\\
$r_j$ & Expected arrival rate of online vertex $j$ with $r_j=T \cdot p_j$.\\
$W_e$ & Random (non-negative) utility on edge $e$ with $\E[W_e]=w_e$.\\
$e$ (Italic) & Edge or assignment $e \in E$.\\
$\sfe$ (Non-italic) & Natural base taking the value around $2.718$. \\
 \hline
\end{tabular}
\end{table}

\section{Benchmark LP}\label{sec:lp}
For an edge $e \in E$, let $x_e$ be the expected number of times that edge $e$ is added in a clairvoyant optimal. For each vertex $j$ ($i$), let $E_j$ ($E_i$) be the set of relevant edges incident to $j$ ($i$). Let $\E[A_{e,k}]=a_{e,k}$ and $\E[W_e]=w_e$ for each $e \in E$ and $k \in [K]$. Our benchmark LP is formally stated as follows.
\begin{alignat}{2}
\max & ~~\sum_{e \in E} w_{e} x_{e} &&  \label{obj-1} \\
 & \sum_{e \in E_j} x_{e} \le r_j  &&~~ \forall j \in J \label{cons:j} \\ 
 &  \sum_{e \in E}  a_{e,k} \cdot x_e\le B_k  && ~~ \forall k \in [K] \label{cons:i} \\ 
  & 0 \le x_{e}   && ~~ \forall e\in E.
  \label{cons:e}
\end{alignat}

Throughout this paper, we refer to the LP above simply as \LP~\eqref{obj-1}. 
\begin{lemma}\label{lem:lp}
The optimal value of \LP~\eqref{obj-1} is a valid upper bound on the total expected utilities achieved by a clairvoyant optimal.
\end{lemma}
\begin{proof}
For each given edge $e=(i,j)$, let $X_e$ be the random number of times that $e$ is added in a clairvoyant optimal (denoted by \OPT) with $x_e=\E[X_e]$. We try to justfy that $\{x_e | e\in E\}$ satisfy all constraints in the \LP above. 

Note that $X_e$ can take values larger than $1$ due to the potential multiple arrivals of the online vertex $j$. However, the total number of edges added with respect to $j$ should be no larger than that of arrivals of $j$ during the online phase, say $R_j$. Thus, $\sum_{e \in E_j} X_j \le R_j$ holds with probability one. Taking expectation on both sides, we get $\sum_{e \in E_j} x_j \le \E[R_j] =r_j$, which leads to Constaint~\eqref{cons:j}. Since we consider safe policies, we are sure that no budget could get vioalted throughout the online process. Thus, for each given resource $k \in [K]$, the event $\sum_{e \in E} A_{e,k}\cdot X_e \le B_k$ occurs with probability $1$. Taking expectation on both sides, we have
\begin{equation}
\E\bb{\sum_{e \in E} A_{e,k}\cdot X_e } =\sum_{e \in E} \E[ A_{e,k}\cdot X_e]
=\sum_{e\in E} \E[X_e]\cdot \E[A_{e,k}]=\sum_{e\in E} x_e\cdot a_{e,k} \le B_k, \label{eqn:lp}
\end{equation}
which yields Constaint~\eqref{cons:i}. Note that since we focus on safe policies, the random realization of $A_{e,k}$ and that of $W_e$ both should be independent of when \OPT adds the edge $e$. This is why equalities in~\eqref{eqn:lp} hold. Moreover, the total expected utilities gained by $\OPT$ is
\[\E\bb{\sum_{e \in E} X_e\cdot W_e}=\sum_{e \in E} \E[X_e] \cdot \E[W_e]=\sum_{e \in E} x_e \cdot w_e, \]
which is consistent with the objective function. Therefore, we conclude that the optimal value to \LP~\eqref{obj-1} should be a valid upper bound for the expected performance of \OPT.
\end{proof}


\section{An LP-based Sampling Algorithm}\label{sec:samp}
Recall that $E_j$ and $E_j$ denote the set of edges incident to $j$ and $i$, respectively. Our first parameterized algorithm \samp is formally stated as follows.
 \begin{algorithm}[ht!]
\DontPrintSemicolon
\textbf{Offline Phase}: \;
Solve  \LP~\eqref{obj-1} and let $\{x^*_{e}|e\in E\}$ be an optimal solution.\;
\textbf{Online Phase}:\;
 \For{$t=1,\ldots,T$}{
Let an online vertex $j$ arrive at time $t$. \label{alg:off-1}  \;
Sample an edge $e \in E_j$ with probability $\alp \cdot x_{e}^*/r_j$. \label{alg:off-3}\;
\eIf{$e$ is safe (\ie there is at least one unit budget of each resource in $\cS_e$); \label{alg:off-5}}{Make the edge $e$.}{Reject $j$.}}
\caption{An LP-based sampling algorithm $\nadap(\alp)$ with $\alp \in [0,1]$.}
\label{alg:nadap}
\end{algorithm}
Note that Step~\eqref{alg:off-3} in $\samp(\alp)$ is valid since $\sum_{e \in E_j} \alp x_e^*/r_j \le \sum_{e \in E_j}  x_e^*/r_j \le 1$ due to Constraint~\eqref{cons:j} of \LP~\eqref{obj-1}.


\subsection{Competitive-ratio  (CR) analysis for $\samp(\alp)$}\label{sec:samp-cr}
In this section, we prove the first part of Theorem~\ref{thm:main-sa}, which states as follows.

\begin{theorem}\label{thm:sa-1}
$\samp(\alp)$ achieves a competitive ratio \emph{equal to} $(1-\sfe^{-\alp \Del})/\Del$ with respect to \LP~\eqref{obj-1}.\end{theorem}

For $\samp(\alp)$, we can re-interpret the online phase as an edge-arriving process such that during each round $t\in [T]$, one edge $e=(i,j)\in E$ will arrive with replacement (\ie $j$ arrives and $e$ gets sampled) with a probability $(r_j/T)\cdot (\alp x_e^*/r_j)=\alp x_e^*/T$. Note that $\sum_{e \in E} \alp x^*_e/T =\sum_{j \in J} \sum_{e\in E_j}  \alp x^*_e/T \le \sum_j r_j/T=1$, where the last inequality is due to Constraint~\eqref{cons:j} in LP~\eqref{obj-1} and $\alp \in [0,1]$. Let $\Z=(Z_k)_{k \in [K]} \in \{0,1\}^K$ be the random consumption of  resources involved in each round of $\samp(\alp)$ when all resources are abundant. By definition, we have $\Z=\cA_e$ with probability $\alp x_e^*/T$ for each $e\in E$. Note that $\E[Z_k]=\sum_{e \in E} \E[A_{e,k}] \cdot (\alp x^*_e/T)=\sum_{e \in E} a_{e,k} \cdot (\alp x^*_e/T)  \le \alp B_k/T$, where the last inequality follows from Constraint~\eqref{cons:i} in LP~\eqref{obj-1}. The proof below for Theorem~\ref{thm:sa-1} will need two lemmas, namely, Lemmas~\ref{lem:supp-1} and~\ref{lem:s-a1}, which proofs are deferred to Appendix. 



\begin{proof}
Consider a given edge $e$ and a given time $t \in [T]$. For each $k \in \cS_{e}$, let $U_{k,t}$ denote the total random cost  of resource $k$ at (the beginning of) $t$ in $\samp(\alp)$. Let $\SF_{e,t}=\bigwedge_{k \in \cS_e} \big(U_{k,t} \le B_k-1\big)$ denote the event that $e$ is safe at $t$. WLOG assume $\Del=|\cS_e|$. Recall that $\Z=(Z_k)_k$ denotes the random consumption of resources involved in each single round of $\samp(\alp)$ when all resources have remaining budgets. For each $k \in [K]$, let $\{Z_{k,t'}| 1\le t' <t\}$ be $t-1$ \iid copies of $Z_k$. Therefore,
\begin{align}\label{ineq:s-a1}
\Pr[\SF_{e,t}]=\Pr\bb{\bigwedge_{k \in \cS_e} \big(U_{k,t} \le B_k-1\big)} \ge \Pr\bb{\bigwedge_{k \in \cS_e} \big(\sum_{1 \le t'<t} Z_{k,t'} \le B_k-1\big)}. 
\end{align}
Observe that for each $k \in \cS_e$, $\E[Z_k]\le \alp B_k/T$. By Lemma~\ref{lem:supp-1}, we see that the right-hand-side value in Inequality~\eqref{ineq:s-a1} gets minimized when $\Z(e):=(Z_k)_{k\in \cS_e}$ are negatively correlated according to the following distribution (denoted by $\cD^*$): with probality $\alp B_k/T$, $\Z(e)=\bo(Z_k)$ for each $k\in \cS_e$,  where $\bo(Z_e)$ refers to the standard basis vector with the only entry being $1$ at the position of $Z_k$, and with probability $1-\sum_{k \in \cS_e} \alp B_k /T$, $\Z(e)=\mathbf{0}$ (a zero vector of length $\Del$). Thus,
\begin{align}\label{ineq:s-a2}
\Pr[\SF_{e,t}] \ge 1-\Pr_{\Z(e) \sim \cD^*}\bb{\exists k \in{\cS_e}: \sum_{1\le t'<t} Z_{k,t'} \ge B_k}.
\end{align}
The value of $\Pr_{\Z(e) \sim \cD^*}\bb{\exists k \in{\cS_e}: \sum_{1\le t'<t} Z_{k,t'} \ge B_k}$ can be interpreted via the following Balls-and-Bins model: There are $t-1$ balls and $\Del=|\cS_e|$ bins; each ball will be thrown independently and it will land in bin $k \in \cS_e$ with a probability $\alp B_k/T$ and land in none of them with probability $1-\sum_{k \in \cS_e} \alp B_k/T$. The value $\Pr_{\Z(e) \sim \cD^*}\bb{\exists k \in{\cS_e}: \sum_{1\le t'<t} Z_{k,t'} \ge B_k}$ then represents the probability that there exists at least one bin $k \in \cS_e$ with at least $B_k$ balls in the end. We claim that $\Pr_{\Z(e) \sim \cD^*}\bb{\exists k \in{\cS_e}: \sum_{1\le t'<t} Z_{k,t'} \ge B_k}$ gets maximized when each $B_k$ takes a value to maximize $\Pr[ \sum_{1\le t'<t} Z_{k,t'} \ge B_k]$: note that the change of $B_k$ only affects the probability of each ball falling into the bin $k$ and  the threshold $B_k$ associated with bin $k$ (and this has nothing to do with the rest of the bins). By Lemma~\ref{lem:s-a1}, we see that $B_k$ should take a value of $1$ for all $k \in \cS_e$. Therefore, in the worst case (when $\Pr[\SF_{e,t}]$ gets minimized), we have (i) $B_k=1$ for $k \in \cS_e$ and (ii) with probability $\alp/T$, $\Z(e)=\bo(Z_k)$ for $k \in \cS_e$, and with probability $1-\Del \alp/T$, all $\Z(e)=\mathbf{0}$. Thus,
 
\begin{align*}\label{ineq:s-a3}
\Pr[\SF_{e,t}] &=\Pr\bb{\bigwedge_{k \in \cS_e} \big(U_{k,t} \le B_k-1\big)} \ge \Pr\bb{\bigwedge_{k \in \cS_e} \big(\sum_{1 \le t'<t} Z_{k,t'} \le B_k-1\big)} \\
&=\Pr\bb{\bigwedge_{k \in \cS_e} \big(\sum_{1 \le t'<t} Z_{k,t'} \le 0\big)}=\Pr\bb{\bigwedge_{k \in \cS_e,t'<t} Z_{k,t'}=0}= \bp{1-\frac{\Del \alp }{T}}^{t-1}.
\end{align*}
Let $M_{e}$ be the total (random) utilities gained on edge $e=(i,j)$ in $\samp(\alp)$. For each $t \in [T]$, let $X_{t}$ indicate if $j$ arrives at $t$ and $Y_{e,t}$ indicate if $e$ is sampled at $t$. Thus, we have
\begin{align*}
\E[M_{e}] &=\sum_{t=1}^T \E[X_t \cdot Y_{e,t} \cdot \SF_{e,t} \cdot W_e]=\sum_{t=1}^T \frac{r_{j}}{T} \cdot \frac{\alp x^*_{e}}{r_{j}}\cdot \Big(1- \frac{\Del \alp}{T}  \Big)^{t-1}\cdot w_e  \\
&\ge x^*_{e} \cdot w_e \cdot \sum_{t=1}^T \frac{\alp}{T} \Big(1- \frac{\Del \alp}{T}  \Big)^{t-1} \ge  x^*_{e}w_{e}\frac{1-\sfe^{-\Del \alp}}{\Del}.
\end{align*}

Thus, we have $\E[M_{e}] \ge x^*_{e} \cdot w_{e} \cdot (1-\sfe^{-\Del \alp})/\Del$. By the linearity of expectation, we claim that the total expected utilities of $\samp(\alp)$ should satisfy
\[
\E[\samp(\alp)] \ge \sum_{e\in E} \E[M_e] \ge \sum_{e\in E} x^*_{e} \cdot w_{e} \cdot (1-\sfe^{-\Del \alp})/\Del=\LP\eqref{obj-1}\cdot (1-\sfe^{-\Del \alp})/\Del \ge \OPT\cdot (1-\sfe^{-\Del \alp})/\Del,
\]
where $\LP\eqref{obj-1}$ denotes the optimal value of the benchmark LP~\eqref{obj-1} and $\OPT$ the total expected utilities achieved by a clairvoyant optimal, and the last inequality above follows from Lemma~\ref{lem:lp}. Thus, we claim that $\samp(\alp)$ achieves a competitive ratio of at least $(1-\sfe^{-\Del \alp})/\Del$. The analysis above actually suggests the tightness of the CR-analysis of $\samp(\alp)$ with respect to LP~\eqref{obj-1}. For completeness, we present an explicit CR worst-case structure on Example~\ref{exam:cr-ws}.  
\end{proof}
\begin{example}[The CR worst-case structure of $\samp(\alp)$]\label{exam:cr-ws}
Consider such an instance of $\mb$ as follows: $|I|=|J|=|E|=1$, $K=\Del$, $B_k=1$ for all $k \in [K]$. The single edge $e$ has such a cost distribution: $\cA_e=\bo_k$ with probability $1/T$ for each $k=1,2,\ldots,K$ and $\cA_e=\mathbf{0}$ with probability $1-K/T$, where $\bo_k$ denotes the $k$th standard basis vector. Also, the edge $e$ has a deterministic unit utility $W_e=1$, and $p_j=1, r_j=T$. We can verify that (i) $x_e^*=T$ in the benchmark LP~\eqref{obj-1} with an optimal value of $T$; (ii) $\samp(\alp)$ gets an expected total utilities equal to $T \cdot (1-\sfe^{-\alp \Del})/\Del$. Thus, we conclude that $\samp(\alp)$ achieves a CR of no more than  $(1-\sfe^{-\alp \Del})/\Del$ with respect to \LP~\eqref{obj-1}. \hfill$\blacksquare$
\end{example}

\subsection{Variance analysis for $\samp(\alp)$ with $\Del=1$}\label{sec:samp-va} 
To better expose our techniques, we start with a simple case of $\Del=1$ here and then go to the general case of $\Del$ in Section~\ref{sec:samp-vb}. 

Note that when $\Del=1$, each edge consumes one single resource only.\footnote{Observe that \mb with $\Del=1$ captures the classical online stochastic matching under KIID as a strictly special case~\cite{feldman2009online,huang2021online,haeupler2011online,manshadi2012online,jaillet2013online,brubach2016new}, when each edge $e=(i,j)$ consumes one single resource of the offline vertex $i$.} For each resource $k \in [K]$, let $E_k=\{e \in E: \cS_e=\{k\}\}$, which denotes the subset of edges whose cost involves the single resource $k$. For notation convenience, we use $a_e$ to denote $a_{e,k}=\E[A_{e,k}]$ for any $e \in E_k$. The  online process of $\samp(\alp)$ can be re-interpreted via the following balls-and-bins model.

\xhdr{An auxiliary balls-and-bins model for variance analysis of $\samp(\alp)$}. We treat each edge as a ball and each edge in $E_k$ is labeled with type $k \in [K]$, and there are $K$ bins and each bin $k$ corresponds to resource $k$. There are $T$ rounds and during each round $t \in [T]$, we sample a ball $e=(i,j) \in E$ with probability $ (r_j/T) \cdot (\alp x_e^*/r_j)= \alp x_e^*/T$ and put it into bin $k$ if $e \in E_k$. Thus, during each round, a ball will be added into bin $k$ with probability $\sum_{e \in E_k} \alp x_e^*/T:=q_k$. Note that $\sum_{k \in [K]} q_k =\sum_{e\in E}\alp x_e^*/T=\sum_{j\in J} \sum_{e\in E_j} \alp x_e^*/T \le \sum_{j\in J} \alp r_j/T= \alp$. Each bin has a capacity $B_k$, and each ball $e \in E_k$ is associated with a Bernoulli random variable of mean $a_{e}$, denoted by $\Ber(a_{e})$. Each time after a ball $e$ is added into bin $k$, the capacity of bin $k$ gets reduced by one with probability $a_e$ and remains unchanged otherwise. \hfill $\blacksquare$

For each $k \in [K]$, let $Y_k$ be the total  (random) number of balls added into bin $k$ by the time when either the capacity $B_k$ is reached or at the end of $t=T$ (whichever comes first). In our context, $\sum_{k \in [K]}Y_k:=Y$ captures the exact total number of edges made in $\samp(\alp)$,  and we aim to upper bound $\Var[Y]$. Note that there are three sources of randomness in $Y$: (\tbf{R1}) the dynamic arrivals of online vetices in $J$;  (\tbf{R2}) the random sampling choices of edges made by $\samp(\alp)$, and  (\tbf{R3}) the random cost realization of $\cA_e$. 
 
\begin{theorem}\label{thm:var-sa}
$\Var[Y] \le T^2 \bp{1 - \sfe^{-2 \alp} - 2 \alp \sfe^{-\alp} +O(1/T)}$. 
\end{theorem}

\begin{proof}
Recall that in each round $t \in [T]$, a ball will be added to bin $k$ with probability $q_k=\sum_{e\in E_k} \alp x_e^*/T$. Let $Y'_k$ be the sum of $T$ \iid copies of $\Ber(q_k)$, which denotes the  total (random) number of balls added into bin $k$ at the end of time $T$ if ignoring the capacity. Observe that (i) $\{Y'_k| k \in [K]\}$ are negatively associated~\cite{dubhashi1996balls}; (2) Each $Y_k$ can be viewed as a non-decreasing function of $Y'_k$. Thus, we claim that $\{Y_k\}$ are also negatively associated~\cite{joag1983negative}. Therefore, by the work of~\cite{shao2000comparison}, we have $\Var[Y]=\Var[\sum_{k\in [K]} Y_k] \le \sum_{k \in [K]} \Var[Y_k]$.

Focus on a given bin $k$. Let $B=B_k$ and $\sum_{e \in E_k} x_e^* \cdot a_e=B \cdot \eta$ with $\eta \in (0,1]$. For each time $t \in [T]$, let $X_t=1$ indicate that one ball  $e \in E_k$ arrives at $t$ and $H_t=1$ indicates that bin $k$ has at least one capacity at (the beginning of) $t$. Thus, $Y_k=\sum_{t=1}^T X_t \cdot H_t$. Observe that (i) $\{X_t| t\in [T]\}$ are $T$ \iid Bernoulli random variables each with mean $q_k=\sum_{e\in E_k} \alp x_e^*/T$; (ii) Assuming bin $k$ has at least one capacity at the beginning of some round, bin $k$ will have one capacity reduced in that round with probability $\sum_{e \in E_k} (\alp x_e^*/T) \cdot a_e = \alp \cdot \eta \cdot B/T$. Therefore, $\E[H_t]=\Pr[\Ber(\alp \eta B/T)^{(t-1)} \le B-1]$, where $\Ber(\alp \eta B/T)^{(t-1)}$ denotes the sum of $t-1$ \iid Bernoulli random variables each with mean $\alp \eta B/T$. Observe that $\E[X_t \cdot H_t]=\E[X_t] \cdot \E[H_t]$ since $H_t$ is independent of $X_t$ for each $t \in [T]$.

\begingroup
\allowdisplaybreaks
\begin{align}
\Var[Y_k] &=\Var\bb{\sum_{t=1}^T X_t \cdot H_t} =\E\bb{\bp{\sum_{t=1}^T X_t \cdot H_t}^2}-\bp{\E\bb{\sum_{t=1}^T X_t \cdot H_t}}^2 \nonumber\\
&=\E\bb{\sum_{t=1}^T X_t \cdot H_t+2\sum_{1\le t'<t \le T} X_t' \cdot X_t \cdot H_{t'} \cdot H_t}-\bp{ q_k \cdot \sum_{t=1}^T \E[H_t]}^2  \nonumber\\
&=q_k \sum_{t=1}^T \E[H_t]+2 q_k^2 \sum_{1 \le t'<t \le T} \E[H_t| X_{t'}=1]-\bp{ q_k \cdot \sum_{t=1}^T \E[H_t]}^2  \nonumber\\
&\le q_k \sum_{t=1}^T \E[H_t]+ q_k^2\cdot\sum_{t=1}^T 2(t-1)\cdot \E[H_t]-\bp{ q_k \cdot \sum_{t=1}^T \E[H_t]}^2 \label{ineq:var-sf}\\
&=q_k \sum_{t=1}^T \E[H_t]+ q_k^2  \bp{\sum_{t=1}^T 2(t-1)\cdot \E[H_t]-\bp{ \sum_{t=1}^T \E[H_t]}^2}  \nonumber
\end{align}
\endgroup
 Inequality~\eqref{ineq:var-sf} follows from the fact that 
$\E[H_t]=\E[H_t|X_{t'}=1]\Pr[X_{t'}=1]+\E[H_t|X_{t'}=0]\Pr[X_{t'}=0] \ge \E[H_t|X_{t'}=1]$ since $\E[H_t|X_{t'}=0] \ge \E[H_t|X_{t'}=1]$. Thus,
\begingroup
\allowdisplaybreaks
\begin{align}
\Var[Y] &\le \sum_{k =1}^K \Var[Y_k]\label{ineq:var-sa1}\\
& \le \sum_{k =1}^K q_k \sum_{t=1}^T \E[H_t]+ \sum_{k =1}^K q_k^2   \bP{\sum_{t=1}^T 2(t-1)\cdot \E[H_t]-\bp{ \sum_{t=1}^T \E[H_t]}^2} \label{ineq:var-sa2}\\
& \le \alp T+\alp^2 \cdot T^2  \cdot \bp{1 - \sfe^{-2 \alp \eta} - 2 \alp \sfe^{-\alp \eta} +O(1/T)} /(\alp \eta)^2\label{ineq:var-sg}\\
&\le T^2 \bp{1 - \sfe^{-2 \alp} - 2 \alp \sfe^{-\alp} +O(1/T)}\label{ineq:var-sh}.
\end{align}
\endgroup
Inequality~\eqref{ineq:var-sg} follows from Inequality $\sum_{k=1}^K q_k^2 \le   \alp^2 $ due to $\sum_{k=1}^K q_k \le \alp$ and Lemma~\ref{lem:var-se} (the proof is deferred to Appendix); Inequality~\eqref{ineq:var-sh} is due to that the function $\bp{1 - \sfe^{-2 \alp \eta} - 2 \alp \sfe^{-\alp \eta}} /(\alp \eta)^2$ is increasing when $\eta \in (0,1]$ for any given $\alp \in [0,1]$.  
\end{proof}

The proof above suggests that the upper bound of variance stated in Theorem~\ref{thm:var-sa} can be \emph{tight}. Observe that Inequality~\eqref{ineq:var-sa1} becomes tight when $K=1$; Inequalities~\eqref{ineq:var-sa2},~\eqref{ineq:var-sg} will be asymptotically tight (after ignoring terms of $O(T)$) when $K=1$, $q_k=\alp$, and $B=1$ (due to Lemma~\ref{lem:var-se}); and Inequality~\eqref{ineq:var-sh} gets tight when $\eta=1$. This reveals the following variance worst-case structure for $\samp(\alp)$ with $\Del=1$.

\begin{example}[A variance worst-case structure of $\samp(\alp)$ with $\Del=1$]\label{exam:va-ws}
Consider such an instance of $\mb$ with $\Del=1$ as follows: $|I|=|J|=|E|=1$, $K=1, B=1$. In other words, there is one single edge and one single resource with a unit budget.
The single edge $e$ has a Bernoulli random cost: $\cA_e=\Ber(1/T)$ with mean $1/T$. We have $p_j=1, r_j=T$. We can verify that (1) $x_e^*=T$ in the benchmark LP~\eqref{obj-1}; (2) $\samp(\alp)$ samples edge $e$ with probability $\alp$ in each round $t\in [T]$; (3) The total (random) number of edges made in $\samp(\alp)$ is equal to $\min(\mathrm{Ge}(\alp/T),T)$, where $\mathrm{Ge}(\alp/T)$ denotes a Geometric random variable of mean $\alp/T$; (4) $\Var[\min(\mathrm{Ge}(\alp/T),T)]= T^2 \bp{1 - \sfe^{-2 \alp} - 2 \alp \sfe^{-\alp} +O(1/T)}$.
 \hfill$\blacksquare$
\end{example}

\subsection{Variance analysis for $\samp(\alp)$ with general $\Del$}\label{sec:samp-vb} 
Recall that in the competitive-ratio analysis of $\samp(\alp)$, we re-interpret the vertex-arriving process as an edge-arriving process such that in each round $t \in [T]$, one edge is  sampled (called \emph{$e$ arrives}) with replacement with probability $(r_j/T)\cdot (\alp x_e^*/r_j)=\alp x_e^*/T:=q_e$ with $\sum_{e\in E} q_e \le \alp$. For each edge $e \in E$ and $t \in [T]$, let $X_{e,t}=1$ indicate that $e$ arrives at $t$ and $\SF_{e,t}=1$ indicate that $e$ is safe at $t$, \ie all resources in $\cS_e$ have at least one unit budget at (the beginning) of $t$. Let $Y=\sum_{e\in E, t\in [T]}X_{e,t} \cdot \SF_{e,t}$, which represents the total (random) number of edges made in $\samp(\alp)$. Thus, our goal is to upper bound $\Var[Y]$. Observe that for each $e=(i,j)$ $\E[X_{e,t}]=q_{e}$ for every $t\in[T]$. 

\begin{theorem}\label{thm:va-sa-ge}
\[
\Var[Y] \le (\alp T)^2 \cdot \sfg\big(\min(\Del \alp, \eta)\big)+O(T),
\]
where $\sfg(x):=\big(1 - \sfe^{-2 x} - 2x \sfe^{-x}\big)/x^2$, and $\eta\sim1.126$ is the unique maximizer of $\sfg(x)$ when $x \in [0,\infty]$ with $\sfg'(x)=0$.
\end{theorem}
\begin{proof}
Note that 
\begin{align*}
\Var[Y]&=\Var\bb{\sum_{e \in E,t\in [T]} X_{e,t}\cdot \SF_{e,t}}\\
&=\E\bb{\bp{\sum_{e \in E,t\in [T]}X_{e,t}\cdot \SF_{e,t} }^2}-\bp{\E\bb{\sum_{e \in E,t\in [T]}X_{e,t}\cdot \SF_{e,t} }}^2:=F_1-F_2,
\end{align*}
where $F_1=\E\bb{\bp{\sum_{e \in E,t\in [T]}X_{e,t}\cdot \SF_{e,t} }^2}$ and $F_2=\bp{\E\bb{\sum_{e \in E,t\in [T]}X_{e,t}\cdot \SF_{e,t} }}^2$. Let $p_{e,t}=\E[\SF_{e,t}]$.
\begin{align*}
F_1&=\underbrace{\sum_{e \in E,t\in [T]} \E[ (X_{e,t}\cdot \SF_{e,t})^2]}_{F_{1,a}}+\underbrace{2 \sum_{e \in E} \sum_{1\le t'<t \le T} \E[X_{e,t'}]\cdot\E[X_{e,t}] \cdot \E[\SF_{e,t}| X_{e,t'}=1]}_{F_{1,b}}+\\
&+\underbrace{2 \sum_{t \in [T]} \sum_{e \neq e'} \E[X_{e,t}\cdot X_{e',t} \cdot \SF_{e,t} \cdot \SF_{e',t}]}_{F_{1,c}} +\underbrace{2\sum_{e' \neq e} \sum_{t'<t} \E[X_{e',t'}\cdot \SF_{e',t'} \cdot X_{e,t}\cdot \SF_{e,t}]}_{F_{1,d}}
\\
&=F_{1,a}+F_{1,b}+F_{1,c}+F_{1,d}.
\end{align*}
We upper bound the four parts of $F_1$ one by one as follows. Let $\sum_{e \in E} q_{e} =\alp' \le \alp$.
\begin{align*}
F_{1,a}&=\sum_{e \in E,t\in [T]} \E[X_{e,t}\cdot \SF_{e,t}] = \sum_{e \in E,t\in [T]} q_e \cdot p_{e,t} \le \alp' \cdot T,\\
F_{1,b} &\le 2 \sum_{e\in E} q_e^2 \sum_{t\in [T]} (t-1) \cdot p_{e,t},\\
F_{1,c}&=0,\\
F_{1,d}&=2\sum_{e'\neq e} \E[X_{e',t'}\cdot X_{e,t}]\sum_{1\le t'<t} \E[\SF_{e',t'} \cdot \SF_{e,t}|X_{e',t'}= X_{e,t}=1]\\
&\le 2\sum_{e\in E} q_e \sum_{e'\neq e} q_{e'} \sum_{t \in [T]} \sum_{1\le t'<t}\E[\SF_{e,t}|X_{e',t'}=1] \le  2\sum_{e\in E} q_e (\alp'-q_e) \sum_{t \in [T]}(t-1) \cdot (p_{e,t}+O(1/T)).
\end{align*}
The analyses for $F_{1,a}$ and $F_{1,b}$ are similar to the previous case of $\Del=1$; For $F_{1,c}$, note that in each round, $\samp(\alp)$  samples at most one single edge and thus, $X_{e,t}\cdot X_{e',t}=0$ for all $t$; As for $F_{1,d}$, observe that in every single round $t'$ with $t'<t$, any arrival of $e'$ with $e'\neq e$ will have a positive impact on the event $\SF_{e,t}$ no more than the fact that no edge arrives in the round $t'$; the latter will have a negligible impact on $\SF_{e,t}$ in terms of at most an additive term of $O(1/T)$. Thus,
\begin{align}
&\Var[Y] =F_1-F_2 \nonumber\\
&\le \alp' \cdot T+2 \sum_{e\in E} q_e^2 \sum_{t\in [T]} (t-1) \cdot p_{e,t}+2\sum_{e\in E} q_e (\alp'-q_e) \sum_{t \in [T]}(t-1) \cdot p_{e,t}-\bp{\E\bb{\sum_{e,t}X_{e,t}\cdot \SF_{e,t} }}^2 \nonumber\\
 &=\alp' \cdot T+\alp' \sum_{e \in E} q_e  \sum_{t\in [T]} 2(t-1) \cdot p_{e,t}-\bp{\sum_{e \in E, t \in [T]} q_e \cdot p_{e,t}}^2. \label{ineq:va-yu1}
\end{align}
The analysis above for $F_{1,d}$ suggests that for any given $\{q_e, p_{e,t} |e \in E, t\in [T]\}$, the variance-WS arrives when all edges share the same support in their random cost distribution, \ie $\{\cS_e|e\in E\}$ are all the same. For this case, we can simply assume there is one single meta edge $e^*$ that arrives with probability $\alp'=\sum_e q_e$ in each round by compressing all the edges into one single edge and together with their cost distributions\footnote{The resulting cost distribution of $\cA_{e^*}$ can be simply expressed as $\cA_{e^*}=\cA_e$ with probability $q_e/\alp'$.}. For each given $t \in [T]$, let $p_t=\Pr[\wedge_{k \in \cS_{e^*}} (\SF_{k,t}=1)]$ be the probability that $e^*$ is safe at $t$. Note that $p_{t} \ge (1-\Del \alp/T)^{t-1} \sim \sfe^{-\Del \alp (t-1)/T}$ following the worst-case competitive-ratio analysis of $\samp(\alp)$ (see Example~\ref{exam:cr-ws}). Meanwhile, $p_{t} \le 1$, which can be asympototically tight in the case when $|\cS_{e^*}|=1$, $\sum_{e \in E} a_{e,k} \cdot x_e^*=B_k \cdot \ep$ for the single resource $k \in \cS_{e^*}$. 
In this case, $p_{t} \ge (1-\alp \ep/T)^T \sim \sfe^{-\alp \ep} \ge 1-\alp \ep$ for all $t \in [T]$. Following this argument, we can simplify the expression on \eqref{ineq:va-yu1} as follows
\begin{align*}
\Var[Y] &\le \alp' \cdot T +{\alp'}^2 \sum_{t\in [T]}2(t-1) \cdot p_t-\bp{\sum_{t \in [T]} \alp' \cdot p_t}^2\\
&=\alp' \cdot T+{\alp'}^2\bb{ \sum_{t\in [T]}2(t-1) \cdot p_t-\bp{\sum_{t \in [T]} p_t}^2} \le \alp \cdot T+\alp^2\bb{ \sum_{t\in [T]}2(t-1) \cdot p_t-\bp{\sum_{t \in [T]} p_t}^2}.
\end{align*}
The analysis on $\bb{ \sum_{t\in [T]}2(t-1) \cdot p_t-\bp{\sum_{t \in [T]} p_t}^2}$ can be reduced to the previous case of $\Del=1$ with an updated value of $\boldsymbol{\alp} \in [0, \Del \alp]$ since $p_t \in [\sfe^{-\Del \alp (t-1)/T},1]$. By Lemma~\ref{lem:var-se}, we have 
\[
\Var[Y] \le \alp \cdot T+\alp^2 \cdot T^2 \cdot \max_{\boldsymbol{\alp} \in [0, \Del \alp]}\sfg(\boldsymbol{\alp})=(\alp T)^2 \cdot\sfg\big(\min(\Del \alp, \eta)\big)+O(T).
\]
The equality above follows from the fact that $\sfg(x):=\big(1 - \sfe^{-2 x} - 2x \sfe^{-x}\big)/x^2$ first increases when $x \in [0, \eta]$ and then decreases when $x \in [\eta,\infty]$, where $\eta \sim 1.126$ is the unique solution of $\sfg'(x)=0$ over $x \in [0,\infty]$.
\end{proof}


\section{A Sampling Algorithm with Time-adaptive Attenuations}\label{sec:att}
Here are the details of simulation-based attenuation in our context.  Suppose an online vertex $j\in J$ arrives at $t$. Consider a given edge $e=(i,j)$. We say $e$ is \emph{safe} at (the beginning of) $t$ iff each resource $k \in \cS_e$ has at least one unit budget at $t$. By applying Monte-Carlo simulations (\ie simulating the online algorithm up to time $t$), we can get a sharp estimate of the probability that each $e$ is safe at $t$ (denoted by $\beta_{e,t}$). Suppose we can show that $\beta_{e,t} \ge \gam_t=(1-\alp \Del/T)^{t-1}$ for all $t\in [T]$, where $\gam_t$ is a time-adaptive targeted value we aim to achieve. By generating an auxiliary Bernoulli random variable $Z_{e,t}$ with mean $\gam_t/\beta_{e,t}$ and adding $Z_{e,t}=1$ as an extra condition for $e$ to be ``safe'' at $t$, we can reduce the probability that each $e$ is ``safe'' at $t$ to be \emph{equal} to our target $\gam_t$.  The formal statement of our algorithm is as follows. Recall that $E_j$ is the set of edges with respect to online vertex $j \in J$. 
\begin{algorithm}[ht!]
\DontPrintSemicolon
\textbf{Offline Phase}: \;
Solve  \LP~\eqref{obj-1} and let $\{x^*_{e}\}$ be an optimal solution.\;
\tcc{In Offline Phase, we aim to compute a sharp estimate of the probability that each edge $e$ is safe at (the beginning of) $t$, denoted by $\beta_{e,t}$.}
\emph{Initialization}: When $t=1$, set $\beta_{e,t}=1$ for all $e \in E$.\;
\For{$t=2,3,\ldots,T$}
{By simulating Step~\eqref{alg:on-1} to Step~\eqref{alg:on-2} of Online Phase for all the rounds $t'=1,2,\ldots,t-1$, we get a sharp estimate of $\beta_{e,t}$, the probability that $e$ is safe at $t$. \label{alg:on-6}\;
}
\textbf{Online Phase}:\;
 \For{$t=1,\ldots,T$}{
Let an online vertex $j$ arrive at time $t$. \label{alg:on-1}  \;
Sample an edge $e \in E_j$ with probability $\alp x_{e}^*/r_j$. \label{alg:on-3}\;
Generate an auxiliary Bernoulli random variable $Z_{e,t}$ with mean $\gam_t/\beta_{e,t}$, where $\gam_t=(1-\alp \Del/T)^{t-1}$.\label{alg:on-4}\;
\eIf{$e$ is safe and $Z_{e,t}=1$\label{alg:on-5}}{Match the edge $e$.}{Reject $j$.}\label{alg:on-2}
}
\caption{A sampling algorithm $\att(\alp)$ with time-adaptive attenuations  with $\alp \in[0,1]$.}
\label{alg:adap}
\end{algorithm}

Throughout this section, we use $\SF_{e,t}=1$ to indicate that $e$ is ``safe'' at $t$, which means that both $e$ is safe at $t$ and $Z_{e,t}=1$ in Step~\eqref{alg:on-5}. By the nature of $\att$, we see $\E[\SF_{e,t}]=\gam_t$ for all $e \in E$ and $t \in [T]$. Note that the condition of $e$ is ``safe'' at $t$ is stricter than that of $e$ is safe at $t$; the latter simply means that every resource $k \in \cS_e$ has at least one unit budget at $t$.


\begin{lemma}\label{lem:adap}
$\beta_{e,t} \ge \gam_{t}=(1-\alp \Del/T)^{t-1}$ for all $e \in E, t \in [T]$. 
  \end{lemma}
\begin{proof}
We prove the lemma by induction on $t \in [T]$. Consider the base case when $t=1$. We have $\beta_{e,t}=1 = \gam_{t}$ for all $e \in E$. Now consider a given $\bt>1$ and assume $\beta_{e,t} \ge \gam_{t}$ and $\E[\SF_{e,t}]=\gam_{t}$ for all $t<\bt$ and all $e\in E$. We show that $\beta_{e, \bt} \ge  \gam_{\bt}$ for all $e \in E$. Consider a given $\be$ and a given $k \in  \cS_{\be}$. We can re-interpret the online vertex-arriving process as edge-arriving alternatively such that each edge $e=(i,j)$  arrives (\ie $j$ arrives and $e$ gets sampled) with probability $(r_j/T) \cdot (\alp x_e^*/T)=\alp x_e^*/T:=q_e$ with $\sum_{e\in E}q_e \le \alp$. Let $U_{k,\bt}$ be the (random) number of consumptions of resource $k$ during the previous $\bt-1$ rounds. For each $t<\bt$, $e \in E$, let $X_{e,t}=1$ indicate that $e$ arrives at $t$. Recall that  $\SF_{e,t}$ be the event that both $e$ is safe at $t$ and $Z_{e,t}=1$ in Step~\eqref{alg:on-5}. Thus, we have 
\begin{align}
\E[U_{k,\bt}]&=\E \Big[\sum_{t<\bt} \sum_{e \in E } X_{e,t} \cdot  \SF_{e,t}\cdot  A_{e,k} \Big] 
=\sum_{t<\bt} \sum_{ e \in E} (\alp x_{e}^*/T)  \cdot (1- \alp \Del/T)^{t-1} \cdot a_{e,k}\nonumber \\
&= \sum_{e  \in E}( x^*_{e} \cdot a_{e,k}) \sum_{t<\bt} (1-\alp \Del/T)^{t-1} \cdot (\alp/T)
 \le B_k \cdot (1/\Del) \cdot\Big(1-\Big(1-\alp \Del/T  \Big)^{\bt-1}\Big). \label{ineq:adap-2}
\end{align}
The inequality~\eqref{ineq:adap-2} is due to Constraint~\eqref{cons:i} of \LP~\eqref{obj-1}.  
Therefore, the probability that $\be$ is safe at $\bt$ before attenuations should be
\begin{align}
\beta_{\be,\bt}&=\Pr\Big[  \bigwedge_{k \in \cS_{\be}} \big(U_{k,\bt} \le B_k-1 \big)\Big] 
\ge 1- \sum_{k \in \cS_{\be}} \Pr[U_{k,\bt} \ge B_k] \label{ineq:adap-33}\\
&\ge 1- \sum_{k \in \cS_{\be}} \frac{E[U_{k,\bt}]}{B_k}  \ge \Big(1- \frac{\alp \Del}{T}  \Big)^{\bt-1}=\gam_{\bt}. \label{ineq:adap-3}
\end{align}
Inequality~\eqref{ineq:adap-33} is  due to union bound; the first Inequality on~\eqref{ineq:adap-3} follows from Markov's inequality while the second is due to Inequality~\eqref{ineq:adap-2} and the fact $|\cS_{\be}| \le \Del$. Therefore, we complete the proof on $\bt$ and finish the induction.
\end{proof}

\xhdr{Proof of the competitive-ratio result of $\att(\alp)$ in Theorem~\ref{thm:main-att}}. 
\begin{proof}
Consider a given edge $e=(i,j)$. Let $M_{e}$ be the total (random) utilities gained on $e$ in $\att(\alp)$. For each $t \in [T]$, let $X_{t}=1$ indicate that $e$ arrives at $t$. Thus, we have
\begin{align*}
\E[M_{e}] &=\sum_{t=1}^T \E[X_t \cdot \SF_{e,t} \cdot W_e]=\sum_{t=1}^T (\alp x_e^*/T)\cdot \Big(1- \frac{\alp\Del}{T}  \Big)^{t-1}\cdot w_e = (x^*_{e}w_{e})\cdot \frac{1-\sfe^{-\alp \Del}}{\Del}.
\end{align*}
The last equality above is obtained by taking $T \rightarrow \infty$. Thus, by the linearity of expectation, we claim that $\adap(\alp)$ achieves an expected sum of utilities equal to $
\E[\att(\alp)]=\sum_{e\in E} (x^*_{e}w_{e})\cdot (1-\sfe^{-\alp \Del})/\Del$, which is a factor of $(1-\sfe^{-\alp \Del})/\Del$ of the optimal value of the benchmark LP~\eqref{obj-1}. Thus, we establish the competitive ratio of $\att(\alp)$.
\end{proof}

\subsection{Variance analysis for $\att(\alp)$}
We re-interpret the online vertex-arriving process as edge-arriving alternatively such that each edge $e=(i,j)$  arrives (\ie $j$ arrives and $e$ gets sampled) with probability $(r_j/T) \cdot (\alp x_e^*/T)=\alp x_e^*/T:=q_e$ with $\sum_{e\in E}q_e \le \alp$.
For each edge $e \in E$ and $t \in [T]$, let $X_{e,t}=1$ indicate that $e$ arrives at $t$. Recall that  $\SF_{e,t}=1$ indicate that both $e$ is safe at $t$ (all resources in $\cS_e$ have at least one unit budget at $t$) and $Z_{e,t}=1$ in Step~\eqref{alg:on-5}. Let $Y=\sum_{e\in E, t\in [T]}X_{e,t} \cdot \SF_{e,t}$, which represents the total (random) number of edges made in $\att(\alp)$. Thus, our goal is to upper bound $\Var[Y]$. Observe that for each $e=(i,j)$ and $t \in [T]$, $\E[X_{e,t}]=q_{e}$ and $\E[\SF_{e,t}]=(1-\alp \Del/T)^{t-1}$. Let $\alp'=\sum_{e\in E}q_e \le \alp$.

\begin{theorem}\label{thm:va-sa-ge}
\[
\Var[Y] \le (\alp T)^2 \cdot \sfg(\alp \Del)+O(T),
\]
where $\sfg(x):=\big(1 - \sfe^{-2 x} - 2x \sfe^{-x}\big)/x^2$.
\end{theorem}
\begin{proof}
Note that 
\[
\mu:=\E[Y]=\sum_{e \in E,t\in [T]} \E[X_{e,t}\cdot \SF_{e,t}]=\sum_{e\in E, t\in [T]} q_e \cdot  \Big(1- \frac{\alp\Del}{T}  \Big)^{t-1}=(\alp'  T) \cdot \frac{1-\sfe^{-\alp \Del}}{\alp \Del}. 
\]
Thus,
\begin{align*}
\Var[Y]&=\Var\bb{\sum_{e \in E,t\in [T]} X_{e,t}\cdot \SF_{e,t}}=\E\bb{\bp{\sum_{e \in E,t\in [T]}X_{e,t}\cdot \SF_{e,t} }^2}-\bp{\E\bb{\sum_{e \in E,t\in [T]}X_{e,t}\cdot \SF_{e,t} }}^2\\
&:=F_{1,a}+F_{1,b}+F_{1,c}+F_{1,d}-\mu^2,
\end{align*}
where
\begin{align*}
F_{1,a}&=\sum_{e \in E,t\in [T]} \E[ (X_{e,t}\cdot \SF_{e,t})^2]=\mu,\\
F_{1,b}&=2 \sum_{e \in E} \sum_{1\le t'<t \le T} \E[X_{e,t'}]\cdot\E[X_{e,t}] \cdot \E[\SF_{e,t}| X_{e,t'}=1] \le 2 \sum_{e \in E} q_e^2 \sum_{t\in [T]} (t-1) \cdot  \Big(1- \frac{\alp\Del}{T}  \Big)^{t-1},\\
F_{1,c} &=2 \sum_{t \in [T]} \sum_{e \neq e'} \E[X_{e,t}\cdot X_{e',t} \cdot \SF_{e,t} \cdot \SF_{e',t}]=0,\\
F_{1,d}&=2\sum_{e' \neq e} \sum_{t'<t} \E[X_{e',t'}\cdot \SF_{e',t'} \cdot X_{e,t}\cdot \SF_{e,t}]\\
& \le 2\sum_{e\in E} q_e (\alp'-q_e) \sum_{t \in [T]}(t-1) \cdot \bp{\Big(1- \frac{\alp\Del}{T}  \Big)^{t-1}+O(1/T)}.
\end{align*}
Summarizing all analyses above, we have
\begin{align*}
\Var[Y] &\le \mu-\mu^2+2 {\alp'}^2 \sum_{t \in [T]}(t-1) \cdot \bp{\Big(1- \frac{\alp\Del}{T}  \Big)^{t-1}+O(1/T)} \\
&\le \mu-\mu^2+2 {\alp}^2 \sum_{t \in [T]}(t-1) \cdot \bp{\Big(1- \frac{\alp\Del}{T}  \Big)^{t-1}+O(1/T)}=(\alp T)^2\cdot \sfg(\alp \Del)+O(T),
\end{align*}
where $\sfg(x):=\big(1 - \sfe^{-2 x} - 2x \sfe^{-x}\big)/x^2$.
\end{proof}

\section{Proof of the hardness result in Theorem~\ref{thm:hard}}\label{sec:hardness}
\begin{example}\label{exam}
Consider the projective plane $\cH=(\cV, \cE)$ of order $\Del-1$ with $\Del-1$ being a prime~\cite{chan2012linear}.  $\cH$ is a hypergraph such that (1) $\cH$ is $\Del$-uniform, $\Del$-regular and intersecting; (2) $|\cV|=|\cE|= \Del^2-\Del+1$; and (3) the natural canonical LP on $\cH$ has an optimal value of $\Del-1+1/\Del$. Now based on $\cH$, we construct an instance of \mb as follows.

Let $G=(I,J,E)$ be a star graph with $|I|=1$, $|J|=|E|=T$, and $r_j=1$ for all $j \in J$. Set $K=\Del^2-\Del+1$, where each offline resource corresponds to one hypervertex. Set $B_k=1$ for all $k \in [K]$. For each hyperedge $f \in \cE$, we create $T/(\Del^2-\Del+1)$ copies of an edge $e \in E$ such that for each edge $e$, (1) with probability $p=(\Del-1+1/\Del)/T$, $W_{e}=1$ and $\cA_{e}=\bfa_f$, where  $\bfa_f \in \{0,1\}^{\Del^2-\Del+1}$ is the canonical representation of $f$; (2) with probability $1-p$, $W_e=0$ and $\cA_e=\bz$, where $\bz$ is the zero vector of dimension $\Del^2-\Del+1$. Thus, we have in total $|E|=(T/(\Del^2-\Del+1))\cdot (\Del^2-\Del+1)=T$ edges. Also,  we see that every edge $e \in E$ has a support of size $\Del$, \ie the number of non-zero entries in $\cA_e$. \hfill $\blacksquare$
\end{example}

\begin{lemma}
 \LP-\eqref{obj-1} has an optimal value at least $\Del-1+1/\Del$ on Example~\ref{exam}.
\end{lemma}

\begin{proof}
For each edge $e$, let $y_e$ be the sum of $x_e$ over its $T/(\Del^2-\Del+1)$ copies. Consider such a solution that $x_e=1$, $y_e=T/(\Del^2-\Del+1)$ for all $e$. We show $\{x_e\}$ is feasible to  \LP-\eqref{obj-1}. First, $\{x_e\}$ is feasible to Constraint~\eqref{cons:j}. Second,  $\{x_e\}$ is feasible to Constraint~\eqref{cons:i}, since for each $k$,  
\[
\sum_{e \in E} x_e \cdot \E[A_{e,k}]=|\{e \in E: \E[A_{e,k}]>0\}| \cdot p=\Del \cdot \frac{T}{\Del^2-\Del+1} \cdot \frac{\Del-1+1/\Del}{T} =1.
\]
Note that for each resource (or hypervertex) $k$, there are exactly $\Del$ different edges $e$ such that each $e$ has $T/(\Del^2-\Del+1)$ copies and each has a non-zero $\E[A_{e,k}]$ with $\E[A_{e,k}]=p$. Therefore, we claim that \LP-\eqref{obj-1} has an optimal value at least $\sum_{e \in E} x_e \cdot p=T \cdot p=\Del-1+1/\Del$.
\end{proof}
Now we start to prove Theorem~\ref{thm:hard}.
\begin{proof}
Consider Example~\ref{exam}. Note that all the $T$ edges are intersecting with each other. We see that the expected total utilities achieved by any online algorithm should be no larger than $1-(1-p)^T=1-\sfe^{-(\Del-1+1/\Del)}$. Thus, we claim that the resulting competitive ratio with respect to \LP~\eqref{obj-1} should be no larger than $\big(1-\sfe^{-(\Del-1+1/\Del)}\big)/(\Del-1+1/\Del)$.
\end{proof}

\section{Proof of Theorem~\ref{thm:main-3}} \label{sec:large}
In this section, we show $\samp(\alp)$ with $\alp=1$  achieves a CR approaching one under the large-budget assumption. WLOG assume that $B_k=B$ for all $k \in [K]$. \emph{Throughout this section, we refer to $\samp(1)$ as $\samp$ for simplicity}.

In the heart of the CR analysis of $\samp$, we need the following key result. 
\begin{theorem}\label{thm:supp-1}
 Assume $B=\omega(\ln \Del)$ and $\Del, B \ll T$. Let $\bfU=(U_k| k\in [\Del]) \sim \cD$, where $\cD$ is a random distribution such that (1) $\bfU \in \{0,1\}^\Del$; and (2) $\E[U_k] \le B/T$ for every $k\in [\Del]:=\{1,2,\ldots,\Del\}$. For each $t \in [T]$, let $\bfU^{(t)}$ be the sum of $t-1$ \iid copies of $\bfU$. Suppose $\cD^*$ is an optimal distribution satisfying the two conditions such that $\sum_{t \in [T]}\Pr_{\bfU\sim \cD^*}\bb{ \max_{k \in [\Del]}U^{(t)}_k \le B-1}$ gets minimized, where $U_k^{(t)}$ is the $k$th entry of $\bfU^{(t)}$. We have that under $\cD^*$, 
  \begin{equation}\label{eqn:supp-1}
  \lim_{T \rightarrow \infty} \frac{1}{T}\sum_{t=1}^T \Pr_{\bfU\sim \cD^*}\bb{\max_{k \in [\Del]} U_k^{(t)} \le B-1}= 1- \kappa \sqrt{\frac{\ln \Del}{B}}(1+o(1)), 
  \end{equation}
   where $\sqrt{2}<\kappa \le 2\sqrt{2}$, and $o(1)$ is a vanishing term when $\Del \rightarrow \infty$.
 \end{theorem}
Now we show how the theorem above partially implies our main Theorem~\ref{thm:main-3}.
\begin{proof}
Consider a given edge $e \in  E$. WLOG assume $\cS_{e}=[\Del] =\{1,2,\ldots, \Del\}$. Let $\bfU=(U_k| k \in[\Del])$ be the random cost of resources in $[\Del]$ during each round when all budgets remain. Following the same analysis in the proof of Theorem~\ref{thm:sa-1}, we have $\E[U_k] \le B/T$ for each $k \in [\Del]$.
For each given $t \in [T]$, let $\cU^{(t)} \in \mathbf{Z}^\Del$ be the total cost of resources in $[\Del]$ during the first $(t-1)$ rounds when running $\nadap$. Observe that (1)  $e$ is safe at $t$ iff the cost of each resource is no larger than $B-1$, \ie $\cU^{(t)} \le (B-1) \cdot \bo$ (entry-wisely and $\bo$ is the vector of ones); (2) $\cU^{(t)} \le \bfU^{(t)}$ (entry-wisely), where $\bfU^{(t)}$ is the sum of $t-1$ \iid copies of $\bfU$. The first is valid since we consider safe policies only, while the second follows that some edge might be unsafe during some previous time $t'<t$. Let $M_{e}$ be the total (random) utilities gained on $e$ in $\nadap$. For each $t \in [T]$, let $X_{e}=1$ indicate that $e$ arrives at $t$ (\ie the online vertex arrives and $e$ gets sampled in $\samp$) and $\SF_t=1$ indicate that $e$ is safe at $t$. Thus, we have
\begin{align*}
\E[M_{e}] &=\sum_{t=1}^T \E[X_e \cdot \SF_{t} \cdot W_e]= x^*_{e} \cdot w_e \cdot \frac{1}{T} \sum_{t=1}^T \E[\SF_t]
=  x^*_{e} \cdot w_e \cdot \frac{1}{T} \sum_{t=1}^T \Pr[\cU^{(t)} \le (B-1) \cdot \bo]\\
& \ge  x^*_{e} \cdot w_e \cdot \frac{1}{T} \sum_{t=1}^T \Pr[\bfU^{(t)} \le (B-1) \cdot \bo]  \ge  x^*_{e} \cdot w_e \cdot  \bp{1- \kappa \sqrt{\frac{\ln \Del}{B}}(1+o(1))}. 
\end{align*}
The last inequality follows from Theorem~\ref{thm:supp-1}. By linearity of expectation, we claim that \nadap achieves a CR at least $1- \kappa \sqrt{\frac{\ln \Del}{B}}(1+o(1))$ with $\kappa \in (\sqrt{2}, 2\sqrt{2}]$. The asymptotically optimality of \nadap can be seen from Theorems~\ref{thm:supp-2} and~\ref{thm:supp-3} for cases of general $\Del$ and $\Del=1$, respectively. Thus, we are done.
\end{proof}
The harness results of Theorem~\ref{thm:main-3} follows from the two theorems below. Due to the space limit, we defer the full proofs of Theorems~\ref{thm:supp-1}, \ref{thm:supp-2} and~\ref{thm:supp-3} to Appendix; see Section~\ref{sec:two-thms}. 
\begin{theorem}\label{thm:supp-2}
No algorithm can achieve a CR asymptotically  better than $1- \kappa \sqrt{\frac{\ln \Del}{B}}(1+o(1))$ when $T \gg B=\omega(\ln \Del)$ and $\Del \gg 1$, where $\kappa$ is the same value as stated in Theorem~\ref{thm:supp-1}. 
\end{theorem}
\begin{theorem}\label{thm:supp-3}
No algorithm can achieve a CR asymptotically  better than $1- \frac{1}{\sqrt{2 \pi B}}(1+o(1))$ when  $T \gg B\gg 1$ and $\Del=1$.
\end{theorem}

\section{Conclusions and Future work}

In this paper, we proposed a model of multi-budgeted online stochastic matching to study assignment problems existing in a wide range of online-matching markets, including online recommendations, rideshares, and crowdsourcing markets. The model features correlated stochastic cost and utility for each assignment. We presented two LP-based parameterized algorithms and analyzed their performance in detail under competitive ratio and variance. 

Our work opens a few research directions. The most urgent one is to close the competitive-ratio gap shown in Figure~\ref{fig:comp}. The second is to identify the exact value of $\kappa$ stated in Theorem~\ref{thm:main-3} under the large-budget assumption. Right now, we can only get an upper and lower bound for it with a gap of $\sqrt{2}$. Additionally, it will be interesting to see if the results in this paper can be generalized to the case of fractional cost values. In other words, can we get similar results if we assume each assignment takes a vector-valued cost from $[0,1]^K$ instead of $\{0,1\}^K$? Lastly, we would like to see if the current variance-variance techniques can be generalized to cope with a more ambitious goal, \ie upper bounding the variance on the total utilities achieved. By previous analyses, we perhaps need to assume utilities on all assignments are upper bounded by some parameter, which is expected to play a critical role in the variance analysis.

 \newpage
\bibliographystyle{unsrtnat}
\bibliography{stable_ref}

\appendix
\onecolumn
\section{Missing proofs in Section~\ref{sec:samp-cr}}
Recall that $\Z(e)=(Z_k)_{k\in \cS_e}$ is a random binary vector of length $\Del$ with $\E[Z_k] \le \alp B_k/T$ for each $k \in \cS_e$. Consider a given $t \in [T]$. Let $Z^{(t)}_k$ be the sum of $t$ \iid copies of $Z_k$ for each $k$. Let $\cD^*$ be such a distribution on $\Z(e)$  that the value $\Pr_{\Z(e) \sim \cD^*}\bb{ \exists k\in \cS_e, {Z_{k}^{(t)}  \ge B_k}}$ gets maximized. 

 \begin{lemma}\label{lem:supp-1}
 $\cD^*$ can always be realized at such a configuration that $\Z(e)=\bo(Z_k)$ with probability $\alp B_k/T$ for each $k \in \cS_e$ and $\Z(e)=\bz$ with probability $1-\sum_{k \in \cS_e} \alp B_k/T $, where $\bo(Z_k)$ refers to the standard basis vector with the only entry being one at the position of $Z_k$, and $\bz$ refers to a zero vector of length $|\cS_e|=\Del$.  \end{lemma}
 
Recall that we assume $T \gg \max_k B_k \ge 1$ and thus, $1-\sum_{k \in \cS_e} \alp B_k/T \ge 0$.
\begin{proof}
We can interpret the value $\Pr_{\Z(e) \sim \cD^*}\bb{ \exists k\in \cS_e, {Z_{k}^{(t)}  \ge B_k}}$ via the following balls-and-bins model: We have $|\cS_e|=\Del$ bins and $t$ balls; during each time $t' \in [t]$, we put a ball into one or multiple bins according to a certain randomized strategy $\cD^*$ such that each bin $k \in \cS_e$ will receive a ball with a marginal probability no more than $\alp B_k/T$; we repeat the process independently for $t$ times and we aim to figure out the strategy $\cD^*$ to maximize the chance that at least one bin $k \in \cS_e$ will get at least $B_k$ balls in the end. 

Suppose under $\cD^*$, with some probability $p \in [0,1]$, we will put a ball into two different bins, say $k$  and $k'$. Consider such a twisted version $\cD$ of $\cD^*$ that we simply modify $\cD^*$ by putting a ball into bin $k$ and bin $k'$ each with a probability $p$ and keep all the rest of $\cD^*$. We can verify that strategy $\cD$ is still feasible in the way that all bins will get a ball with a marginal probability the same as before. Consider the two materialization trees of $\cD^*$  and $\cD$ where each has a depth of $t$. We say a path is successful if it ends with that at least a bin $i \in \cS_e$ has at least $B_i$ balls. We can construct an \emph{injective} mapping between all successful paths on the tree of $\cD^*$ and those on tree of $\cD$. This suggests that the chance of success under strategy $\cD^*$ should be no larger than that under $\cD$. Keeping on arguing in the way above, we get our claim. 
\end{proof}

\begin{lemma}\label{lem:s-a1}
Let $Z^{(t)}$ be the sum of $t$ \iid Bernoulli random variables each has a mean of $\alp B/T$ with $\alp \in [0,1]$, $1 \le B \ll T$, and $1\le t \le T$.  We have that $\Pr[Z^{(t)} \ge B]$ get maximized when $B=1$.
\end{lemma}
\begin{proof}
It  suffices to show that $\Pr[Z^{(t)} \le B-1]$ will get minimized at $B=1$.
Note that $\E[Z^{(t)}]=\alp t B/T:=\mu$. 

When $\mu <B-1$, by Lemma 11 of~\cite{baveja2018improved}, we have $\Pr[Z^{(t)} \le B-1] \ge \Pr[\Pois(\mu) \le B-1]$, where $\Pois(\mu)$ denotes a Poisson random variable with mean $\mu$. When $\mu \ge B-1$,  we can view that $Z^{(t)} \sim \Pois(\mu)$ since  $1\le B \ll T$ (here we can treat $B$ is a given constant since $T \gg B$). Thus, we claim that $\Pr[Z^{(t)} \le B-1] \ge \Pr[\Pois(\mu) \le B-1]$ with $\mu=\alp tB/T \le B$~\footnote{More precisely, we have $\Pr[Z^{(t)} \le B-1] \ge \Pr[\Pois(\mu) \le B-1]-O(1/T)$. We ignore terms of $O(1/T)$ since $T \gg B \ge 1$.}. Now we show that for any given $\tau:=\alp t/T \in [0,1]$, $f(B):=\Pr[\Pois(\tau B) \le B-1]$ is an increasing function on $B=1,2,\ldots$. Observe that 
\begin{align}
f(B) \cdot  \sfe^{\tau B+\tau}&=\sfe^{\tau}\sum_{\ell=0}^{B-1} (\tau B)^{\ell}/\ell!=\bp{\sum_{\ell=0}^\infty \tau^\ell/\ell!} \cdot \bp{ \sum_{\ell=0}^{B-1} (\tau B)^{\ell}/\ell!} \label{eqn:s-a1}:=g_B(\tau)\\
f(B+1) \cdot\sfe^{\tau B+\tau} &= \sum_{\ell=0}^{B} (\tau B+\tau)^{\ell}/\ell! :=g_{B+1}(\tau)\label{eqn:s-a2}
\end{align}
Suppose we regard the right-hand-side of Equations~\eqref{eqn:s-a1} and~\eqref{eqn:s-a2} both as  polynomials of $\tau$, denoted by $g_B(\tau)$ and $g_{B+1}(\tau)$, respectively. We can verify the two share the same coefficients of $\tau^\ell$ for all $0 \le \ell \le B-1$ and further more, we can verify that $g_B(\tau)-g_{B+1}(\tau)=\tau^B (\sum_{\ell=1}^\infty c_\ell \tau^\ell-B^B/B!)$ with $c_\ell >0$. Observe that 
\[
f(B) \le f(B+1) \Leftrightarrow f(B) \cdot\sfe^{\tau B+\tau}  \le f(B+1) \cdot  \sfe^{\tau B+\tau}\Leftrightarrow g_B(\tau)-g_{B+1}(\tau) \le 0 \Leftrightarrow \sum_{\ell=1}^\infty c_\ell \tau^\ell-B^B/B! \le 0.
\]
Thus, it would suffice to show the case when $\tau=1$. In this case, $f(B)=\Pr[\Pois(B) \le B-1]$ is increasing on $B$, which follows from Lemma 1 of~\cite{adell2005median}. 

Summarizing all analyses above, we claim $\Pr[Z^{(t)} \le B-1] \ge  \Pr[\Pois(\alp t B/T) \le B-1] \ge \Pr[\Pois(\alp t/T) \le 0]=\sfe^{-\alp t/T}$. Observe that when $B=1$, $\Pr[Z^{(t)} \le B-1]=(1-\alp /T)^t \sim \sfe^{-\alp t/T}$ when $T \gg 1$. Therefore, we establish that $\Pr[Z^{(t)} \le B-1]$ gets minimized at $B=1$. 
\end{proof}

\section{Missing proofs in Section~\ref{sec:samp-va}}
\begin{lemma}\label{lem:var-se}
Let $H_t=1$ indicate that $\Ber(\alp B/T)^{(t-1)} \le B-1$ and $H_t=0$ otherwise, where $\Ber(\alp B/T)^{(t-1)}$ is the sum of $t-1$ \iid Bernoulli random variables each with mean $\alp B/T$. We have that for any given $\alp>0$,
\[
F_\alp(B):=\sum_{t=1}^T 2(t-1)\cdot \E[H_t]-\bp{ \sum_{t=1}^T \E[H_t]}^2 \le \frac{T^2}{\alp^2} \bp{1 - \sfe^{-2 \alp} - 2 \alp \sfe^{-\alp}+O(1/T)},
\]
where the inequality above becomes tight when $B=1$.
\end{lemma}

\begin{proof}
For notation convenience, let $Z^{(t-1)}:=\Ber(\alp B/T)^{(t-1)}$, which denotes the sum of $t-1$ \iid Bernoulli random variables each with mean $\alp B/T$.  Let $p_t=\E[H_t]=\Pr[Z^{(t-1)} \le B-1]$.  Consider a given $t$ and suppose we try to view $F_\alp(B)$ as a function of $p_t$ and all $\{p_{t'}| t' \neq t\}$ are constants. We can verify that 
\[
F_\alp(B)=-p_t^2+2p_t\bp{t-1- \sum_{t' \neq t} p_{t'}}+C:=f(p_t), 
\]
where $C$ is a constant (a function of $\{p_{t'}| t' \neq t\}$). Note that by Lemma~\ref{lem:s-a1}, $p_t=\Pr[Z^{(t-1)} \le B-1]$  gets minimized at $B=1$ and thus, $p_t \ge (1-\alp /T)^{t-1}$. Therefore, we can verify that $\sum_{t'\neq t} p_{t'} \ge \Omega(T)$, which suggests that for any given $t \le  \Omega(T)$, $f(p_t)$ is decreasing when $p_t \in [0,1]$. In the proof of Lemma~\ref{lem:s-a1}, we show $p_t=\Pr[Z^{(t-1)} \le B-1] \ge\Pr[ \Pois((t-1)\alp B/T) \le B-1]:=\ell_t$ for any $t \in [T]$. This implies that $\Var[H]$  will get non-decreased if we replace $p_t$ with  $\ell_t$ for any $t \le (1-1/\sfe) T$. Note that $p_t=\ell_t$ for any $t \ge \Omega(T)$.\footnote{For analysis convenience, we ignore all $O(1/T)$ terms involved here since $T \gg B \ge 1$.} Thus,

\begingroup
\allowdisplaybreaks
\begin{align}
&F_\alp(B) \le  \sum_{t=1}^T 2(t-1)\cdot \ell_t-\bp{\sum_{t=1}^T \ell_t}^2\\
&\le  \sum_{t=0}^{T-1} 2t \cdot \sfe^{-\alp B t/T} \sum_{\ell=0}^{B-1} \frac{(\alp B t/T)^\ell}{\ell!}-\bP{\sum_{t=0}^{T-1} \sfe^{-\alp B t/T} \sum_{\ell=0}^{B-1} \frac{(\alp B t/T)^\ell}{\ell!}}^2\\
&=\sum_{\ell=0}^{B-1} T^2 \sum_{t=0}^{T-1} \frac{1}{T} \frac{2t}{T}\sfe^{-\alp B t/T}\frac{(\alp B t/T)^\ell}{\ell!}-\bP{\sum_{\ell=0}^{B-1} T \sum_{t=0}^{T-1} \frac{1}{T} \sfe^{-\alp B t/T}\frac{(\alp B t/T)^\ell}{\ell!}}^2\\
&=\sum_{\ell=0}^{B-1} T^2\bP{\int_0^1 2\zeta \cdot \sfe^{-\alp B \zeta} \frac{(\alp B \zeta)^\ell}{\ell!} d\zeta +O(1/T)}-\bB{\sum_{\ell=0}^{B-1} T \bP{\int_0^1 \sfe^{-\alp B \zeta} \frac{(\alp B \zeta)^\ell}{\ell!} d\zeta +O(1/T)}}^2 \\
&=T^2 \bB{\int_0^1 2 \zeta \cdot \Pr\bB{\Pois(\alp B \zeta) \le B-1} d\zeta-\bp{\int_0^1   \Pr\bB{\Pois(\alp B \zeta) \le B-1} d\zeta}^2 }+o(T^2). \label{ineq:var-sc}
\end{align}
\endgroup
We show that the expression on~\eqref{ineq:var-sc}  gets maximized at $B=1$ any given constant $\alp>0$ (Actually, we can show it is strictly increasing when $B=1,2,\ldots$ for any given $\alp>0$). For the convenience of exposition,  we just show the case when $\alp=1$~\footnote{For a general case $\alp>0$, we can reduce to the case $\alp=1$ by treating $Z^{(t-1)}$ as $t-1$ \iid $\Ber(B/T')$ with $T'=T/\alp$.}.
\begin{align}
f(B)&:=\int_0^1 2 \zeta \cdot \Pr\bB{\Pois(B \zeta) \le B-1} d\zeta-\bp{\int_0^1   \Pr\bB{\Pois(B \zeta) \le B-1} d\zeta}^2  \\
& \le \int_0^1 2 \zeta  d\zeta-\bp{1-\int_0^1   \Pr\bB{\Pois(B \zeta) \ge B} d\zeta}^2 \\
& \le 1-\bP{1- \int_0^1  \sfe^{-\zeta^2 B/2}d\zeta}^2  =1-\bP{1- \frac{1}{\sqrt{B}}\int_0^{\sqrt{B}}  \sfe^{-\zeta^2 /2} d\zeta}^2 \label{ineq:var-se}  \\
& \le 1-\bp{1- \frac{1}{\sqrt{B}}\int_0^\infty  \sfe^{-\zeta^2 /2}d\zeta}^2=\sqrt{\frac{\pi}{2B}} \bp{2+\sqrt{\frac{\pi}{2B}}},
\end{align}
where Inequality~\eqref{ineq:var-se} is due to the upper-tail bound of a Poisson random variable~\citep{pois-tail}. We can manually verify that $f(B)$ is strictly decreasing when $B=1,2,\ldots,1000$ with $f(1)=1-\sfe^{-2}-2\sfe^{-1}\sim 0.1289$, and  $f(B)  \le 0.08 $ when $B >1000$ by the upper bound shown above. Thus, we claim that $f(B)$ gets maximized at $B=1$.  

Substituting $B=1$ back to Inequality~\eqref{ineq:var-sc}, we have
\begin{align}
F_\alp (B)&\le T^2 \bB{\int_0^1 2 \zeta \cdot \Pr\bB{\Pois(\alp  \zeta) \le 0} d\zeta-\bp{\int_0^1   \Pr\bB{\Pois(\alp  \zeta) \le 0} d\zeta}^2 }+o(T^2) \\
&=T^2 \bp{1 - \sfe^{-2 \alp} - 2 \alp \sfe^{-\alp}}/\alp^2+o(T^2).
\end{align}
\end{proof}


\section{Proofs of Theorems~\ref{thm:supp-1},~\ref{thm:supp-2}, and~\ref{thm:supp-3}} \label{sec:two-thms}

\subsection{Proof of Theorem~\ref{thm:supp-1}}\label{sec:supp-1}
From Lemma~\ref{lem:supp-1}, we see that the optimal distribution $\cD^*$ stated in Theorem~\ref{thm:supp-1} can always be realized at such a configuration that $\bfU=\bo_k$ with probability $B/T$ for each $k \in [\Del]$ and $\bfU=\bz$ with probability $1- \Del \cdot B/T$, where $\bo_k$ denotes the $k$th standard basis vector. Under this setting, the question presented in Theorem~\ref{thm:supp-1} can be formulated as a Balls-and-Bins problem as follows. 

\xhdr{A Balls-and-Bins Problem (BBP)}. Let $\Del, B$ and $T$ are three integers with $\Del \le  B \ll T$. Suppose we have $\Del$ bins and $T$ rounds. In each round $t \in [T]$, we throw a ball such that it will hit each bin with probability $B/T$ and it will miss all the bins with probability $1-\Del  B/T$. Let $\bfU =(U_k| k\in [\Del])$ be the random vector taking values from $\{0,1\}^\Del$, which captures the numbers of balls falling into the $\Del$ bins in one single round. Observe that $\bfU$ follows the exact distribution of $
\cD^*$ as stated in Lemma~\ref{lem:supp-1}. Let $\widehat{T} \le T$ and $T' \le T$ be the last rounds such that every bin has at most $B-1$ balls before and after the arrival of ball, respectively. We see that $T' \le \widehat{T} \le T'+1$, and

\[
 \lim_{T \rightarrow \infty} \frac{1}{T}\sum_{t=1}^T \Pr_{\bfU\sim \cD^*}\bb{\max_{k \in [\Del]} U_k^{(t)} \le B-1}=\lim_{T \rightarrow \infty}\frac{\E[\widehat{T}]}{T} =\lim_{T \rightarrow \infty}\frac{\E[T']}{T}\doteq \ra.
\]

\xhdr{An upper bound on $\ra$}.
\begin{lemma}\label{lem:upper}
$\ra \le 1-(\sqrt{2}-\ep)\sqrt{\frac{ \ln \Del }{B}} (1+o(1))$ for any given $\ep >0$, where $o(1)$ is a vanishing term when $\Del$ approaches infinity.
\end{lemma}

To prove the above lemma, we need Slud's Inequality~\citep{slud1977distribution} stated as follows.

\begin{lemma}[Slud's inequality~\citep{slud1977distribution}]\label{lem:slud}
Let $\{X_i| i \in [n]\}$ be $n$ \emph{i.i.d.} Bernoulli random variables with $\E[X_i]=p$ and let $k \le n$ be a given integer. If either (a) $p \le 1/4$ and $np \le k$, or (b) $np \le k \le n (1-p)$, then 
\[
\Pr\bb{\sum_{i=1}^n X_i  \le k-1}  \le \Phi \bp{\frac{k-np}{\sqrt{np(1-p)}}},\] 
where $\Phi$ is the cdf of a standard normal distribution.
\end{lemma}

Here is an explanation for the above inequality. When $n$ is large, we can approximation the distribribution of $\sum_i X_i$ by $\cN(np, np(1-p))$. Thus, we see that $\Pr[\sum_i X_i \le k-1] \sim \Phi((k-1-npt)/\sqrt{np(1-p)})$. 
Now we start to prove Lemma~\ref{lem:upper}.
\begin{proof}
For each time $t \in [T]$ and bin $k \in [\Del]$, let $U_k^{(t)}$ be the number of balls in the $i$th bin at the end of time $t$. Thus, we have 
\begin{align}
\E[T'] &=\sum_{t=1}^T \Pr\bb{ \bigwedge_{i \in [\Del]}(U_k^{(t)} \le B-1)} \le \sum_{t=1}^T\prod_{i \in [\Del]} \Pr[U_k^{(t)} \le B-1] \label{ineq:up1}  \\
& \le  B-1+\sum_{t=B}^T \Phi^\Del  \bp{\frac{B-tB/T}{\sqrt{t (B/T)(1-B/T)}}}.\label{ineq:up2}
\end{align}

The first inequality~\eqref{ineq:up1} is due to the fact that $\{U^{(t)}_k| i \in [\Del]\}$ are \emph{negatively associated}~\citep{joag1983negative}. The second inequality~\eqref{ineq:up2} is due to Slud's inequality. Observe that when $t \le B-1$,  $\Pr[U^{(t)}_k \le B-1]=1$ for every $i \in [\Del]$. When $t \ge B$, we can view each $U^{(t)}_k$ as a sum of $t$ \iid Bernoulli random variable with mean $B/T$, and verify that Condition (a) stated in Lemma~\ref{lem:slud} is satisfied since $B/T \le 1/4$ and $t B/T \le B$. Therefore, 

\begin{align}
\ra &=\lim_{T \rightarrow \infty} \frac{\E[T']}{T} \nonumber \\
& \le \lim_{T \rightarrow \infty}\bP{ \frac{B-1}{T}+\frac{1}{T}  \sum_{t=B}^T \Phi^\Del  \bp{\frac{B-tB/T}{\sqrt{t (B/T)(1-B/T)}}}} = \int_0^1 dz \Phi^\Del\bp{\frac{\sqrt{B}(1-z)}{\sqrt{z}}} \nonumber\\
&= \int_0^{1-\del} dz \Phi^\Del\bp{\frac{\sqrt{B}(1-z)}{\sqrt{z}}}+  \int_{1-\del}^{1} dz \Phi^\Del\bp{\frac{\sqrt{B}(1-z)}{\sqrt{z}}}\nonumber \\
& \le  (1-\del)+\del \cdot \Phi^\Del \big(\sqrt{B} \del /\sqrt{1-\del}\big) \label{ineq:up3}.
\end{align}
Inequality~\eqref{ineq:up3} follows from the fact that $f(z) \doteq \Phi^\Del\bp{\sqrt{B}(1-z)/\sqrt{z}}$ is an non-increasing function over $z \in [0,1]$. Set $\del=c\sqrt{\ln \Del /B}$ with $c=\sqrt{2}-\ep$. In the following, we prove that $\Phi^\Del \big(\sqrt{B} \del /\sqrt{1-\del}\big)=o(1)$ when $\Del \rightarrow \infty$ for any given $\ep>0$. This yields our main result.

Let $Q(x)=1-\Phi(x)$. Observe that when $\del=c\sqrt{\ln \Del /B}$ with $c=\sqrt{2}-\ep$,
\begin{align}
\Phi^\Del \big(\sqrt{B} \del /\sqrt{1-\del}\big) &=\Phi^\Del \bp{c \sqrt{\ln \Del}/\sqrt{1-\del}} 
\\
&=\bp{1-Q\big(c  \sqrt{\ln \Del}/\sqrt{1-\del}\big)}^\Del \label{ineq:up4}\\
&\le  \bp{1-\frac{\sfe^{-c^2 \ln \Del/(2(1-\del))}}{\sqrt{2\pi} \big( 2c\sqrt{\ln \Del}\big)}}^\Del \label{ineq:up5}\\
& \le  \exp\bp{-\Del^{1-\frac{c^2}{2(1-\del)}} \frac{1}{  2\sqrt{2\pi}\cdot c \cdot \sqrt{\ln \Del}}}=o(1). 
\end{align}

Inequality~\eqref{ineq:up5} follows from the fact that $Q(x) \ge \phi(x)/(x+1/x)$ for $x>0$ with $\phi(x)$ being the pdf of the standard norm distribution~\citep{borjesson1979simple}. Thus, we claim that $\Phi^\Del \big(\sqrt{B} \del /\sqrt{1-\del}\big) $ is a vanishing term when  $\Del \rightarrow \infty$ with $\del=c\sqrt{\ln \Del /B}=o(1)$ since $B=\omega(\ln \Del)$ and $c<\sqrt{2}$.
\end{proof}

\xhdr{A lower bound on $\ra$}.
\begin{lemma}\label{lem:lower}
$\ra \ge 1-2\sqrt{2}\sqrt{\frac{ \ln \Del }{B}} (1+o(1))$ for any given $\ep >0$, where $o(1)$ is a vanishing term when $\Del$ approaches infinity.
\end{lemma}

 To prove the above lemma, we need the following key lemma, which appears as Lemma 5.10 on page 110 of the book~\citep{mitzenmacher2017probability}.
 
\begin{lemma}[\cite{mitzenmacher2017probability}]\label{lem:mit}
Consider a balls-and-bins model with $m$ balls and $n$ bins. Let $Y_i^{(m)}$ be the number of balls in the $i$th bin at the end of $m$ rounds. let $\{Z_i^{(m)}| i \in [n]\}$ be $n$ \iid Poisson random variables each with mean $m/n$. Let $f(x_1,\ldots,x_n)$ be a nonnegative function such that $\E[f(Y_1^{(m)},\ldots,Y_n^{(m)} )]$ is either monotonically increasing or monotonically decreasing in $m$. Then
\[\E[f(Y_1^{(m)},\ldots,Y_n^{(m)} )] \le 2 \E[f(Z_1^{(m)},\ldots,Z_n^{(m)} )].\]
\end{lemma}
Now we start to prove Lemma~\ref{lem:lower}.
\begin{proof} Recall that $\ra=\lim_{T \rightarrow\infty}\E[T']/T$, where $T'$ is the last round when all bins have at most $B-1$ balls. 
\[
\E[T']=\sum_{t=1}^T\Pr[\max_k U_k^{(t)} \le B-1]=\sum_{t=1}^T \Big(1- \Pr[\max_i U_k^{(t)}  \ge B] \Big)=T -\sum_{t=1}^T \Pr[\max_k U_k^{(t)} \ge B].
\]

Therefore, 
\begin{align}
\E[T']/T &=1-\frac{1}{T}\sum_{t=1}^T \Pr[\max_k U_k^{(t)} \ge B] 
=1-\frac{1}{T} \sum_{t=B}^T  \Pr[\max_k U_k^{(t)} \ge B] \\
& \ge 1-\frac{2}{T} \sum_{t=B}^T  \Pr[\max_k Z_k^{(t)} \ge B] ~~\mbox{\bp{$\{Z_k^{(t)}\}$ are $\Del$ \iid Poisson r.v.s each with mean $tB/T$}} \label{ineq:1}\\
&= 1-\frac{2}{T} \sum_{t=B}^T \bB{1- \Big(\Pr[Z_k^{(t)} \le B-1] \Big)^\Del } \\
 & = -1+O(B/T)+\frac{2}{T}\sum_{t=1}^T \Big(\Pr\Big[ \mathrm{Pois}(t B/T) \le B-1 \Big]\Big)^\Del.   \label{ineq:lb-7}
  \end{align}

Inequality~\eqref{ineq:1} is due to Lemma~\ref{lem:mit}. In our case, we have $n=T/B$ and $m=t$, and $f\big(Y_1^{(m)},\ldots,Y_n^{(m)}\big) \doteq\bo\bp{U_k^{(t)} \ge B, \forall 1\le k \le \Del}$, which is an indicator function showing if one of the first $\Del$ bins has at least $B$ balls by the round of $t$. We can verify that $\E[f(Y_1^{(t)},\ldots,Y_n^{(t)} )]$ is monotonically increasing in $t$ under our definition.

Thus, Inequality~\eqref{ineq:lb-7} suggests that
\begin{align}
\ra &=\lim_{T \rightarrow \infty} \frac{\E[T']}{T}\ge  -1+2 \int_0^1 dx \Big(\Pr\Big[ \mathrm{Pois}(x B) \le B-1 \Big]\Big)^\Del   \doteq 1-2 F, \label{ineq:lb-6}
\end{align}
where $F= 1-\int_0^1 dx \bp{1-\Pr[\Pois(x B) \ge B]}^\Del$. We try to upper bound $F$ as follows. Note that,
\begin{align}
&F  \le  \int_0^1 dx \bB{1-\bP{1-\exp \bp{\frac{-B}{2}x^2}}^\Del} \label{ineq:lb-a-1}\\
&=\sqrt{\frac{2}{B}} \int_0^{\sqrt{\frac{B}{2}}}d\theta \bB{1-\bP{1-\sfe^{-\theta^2}}^\Del}\le \sqrt{\frac{2}{B}} \int_0^{\infty}d\theta \bB{1-\bP{1-\sfe^{-\theta^2}}^\Del}\\
&\le \sqrt{\frac{2}{B}} \bB{ \sqrt{\ln \Del} +
 \int_{ \sqrt{\ln \Del}}^{\infty}d\theta \bB{1-\bP{1-\sfe^{-\theta^2}}^\Del} }
 \label{ineq:lb-a-3}\\
& \le  \sqrt{\frac{2}{B} } \bB{\sqrt{\ln \Del}+\Del \int_{  \sqrt{\ln \Del}}^{\infty}d\theta  ~\sfe^{-\theta^2}} \label{ineq:lb-a-4}\\
& \le \sqrt{\frac{2}{B}}  \bB{\sqrt{\ln \Del}+\Del \cdot \frac{1}{2 \sqrt{\ln \Del}} \cdot \sfe^{-( \sqrt{\ln \Del})^2}}=\sqrt{2}\sqrt{\frac{\ln \Del}{B}}(1+o(1)). 
\label{ineq:lb-a-5}
\end{align}

Inequality~\eqref{ineq:lb-a-1} is due to the upper tail bound of a Poisson random variable~\citep{pois-tail}. Inequality~\eqref{ineq:lb-a-4} is due to the fact that $f(\theta)\doteq\bB{1-\bP{1-\sfe^{-\theta^2}}^\Del} \le \Del \cdot \sfe^{-\theta^2}$. Inequality~\eqref{ineq:lb-a-5} follows from the fact that $Q(x)=1-\Phi(x) \le \phi(x)/x$ for any $x>0$ with $\phi(x)$ being the density function of the standard normal distribution~\citep{borjesson1979simple}.

Plugging the result of Inequality~\eqref{ineq:lb-a-5} back to \eqref{ineq:lb-6}, we get that
\[\ra=\lim_{T \rightarrow \infty}\frac{\E[T']}{T} \ge 1- 2\sqrt{2}\sqrt{\frac{\ln \Del}{B}}(1+o(1)).
\]
\end{proof}

\subsection{Proof of Theorem~\ref{thm:supp-2}}

Based on the structure shown in Lemma~\ref{lem:supp-1}, we construct the worst scenario of \nadap as follows.

\begin{example} \label{exam:supp-2}
Consider such a simple case that $|I|=|J|=1$ and $q_j=1$ and $r_j=T$ where $T$ is the total number of online rounds. In other words, there is one single online agent which arrives with probability one during each round $t \in [T]$, and there is one single edge, denoted by $e$. Assume $W_e =1$ with probability one. Let $K=\Del$ and $\cA_e=\bo_k$ with probability $B/T$ for each $k \in [\Del]$ and $\cA_e=\bz$ with probability $1-\frac{\Del B}{T}$, where $\bo_k \in \{0,1\}^\Del$ refers to the $k$th standard basis vector. Thus, by definition we have $\cS_e=[\Del]$ with $|\cS_e|=\Del$. Set $B_k=B$ for all $k \in [K]$.  $\hfill \square$
\end{example}

\begin{lemma}\label{lem:supp-2}
No algorithm can achieve an online competitive ratio better than $1- \kappa \sqrt{\frac{\ln \Del}{B}}(1+o(1))$ on Example~\ref{exam:supp-2}, where $\kappa$ is the value as stated in Theorem~\ref{thm:supp-1} and $o(1)$ vanishes when $B= \omega(\ln \Del)$ and  $\Del \rightarrow \infty$. 
\end{lemma}
\begin{proof}
We can verify that the optimal value to  \LP-\eqref{obj-1} is $T$ since $x_e=T$ is the unique optimal solution. Note that there is one single edge $e$, which arrives during each round with probability one. Thus, the expected total utilities achieved by any policy should be no greater than that by Greedy, which will add the edge $e$ whenever $e$ is safe. Observe that $e$ is safe iff the consumption of each resource $k$ is no larger than $B-1$. Let $T' \le T$ be the last round such that the consumption of every resource $k \in [K]$ is no larger than $B-1$ (before the arrival of online agent). We see that 
 $\gre \le \E[T']$, which suggests that the CR of any algorithm will never beat $\E[T']/T$ with respect to \LP-\eqref{obj-1}. Note that if applying \gre to Example~\ref{exam:supp-2}, the random resource consumption can be captured by the Bins-and-Balls model in Section~\ref{sec:supp-1}. Particularly, $T'$ is exactly equal to the last round that all bins have at most $B-1$ balls. Therefore, we claim that any policy should achieve a CR no larger than $1- \kappa \sqrt{\frac{\ln \Del}{B}}(1+o(1))$ on Example~\ref{exam:supp-2}. \end{proof}

\subsection{Proof of Theorem~\ref{thm:supp-3}}\label{sec:proof-part-a}


\begin{lemma}\label{lem:supp-3}
When $\Del=1$,
$ \lim_{T \rightarrow \infty}\E[\widehat{T}]/T=1-\frac{1}{\sqrt{2 \pi B}}(1+o(1))$, where $o(1)$ is a vanishing term when $B \rightarrow \infty$.
\end{lemma}

\begin{proof}
Consider a special case of BPP stated in Section~\ref{sec:supp-1} with $\Del=1$. In this case, there is only one bin and during each round $t \in [T]$, it will get a ball with probability $B/T$ and nothing otherwise. Recall that $\widehat{T} \le T$ is the last round such that the bin has at most $B-1$ balls before the arrival of ball. We can recast $\widehat{T} $ as the smallest integer $n \le T$ such that $ \sum_{i=1}^n X_i =B$ or $\widehat{T} =T$   where  $\{X_i\}$ are $\iid$ Bernoulli random variables each  with mean $B/T$. Let $X=\sum_{i=1}^{\widehat{T}} X_i$. Note that $\widehat{T} $ qualified as a \emph{stopping} time. By Wald's equation, we have $\E[X]=\E[\widehat{T} ] \cdot \E[X_i]$.  Thus,
\[
 \lim_{T \rightarrow \infty}\frac{\E[\widehat{T}]}{T}=\lim_{T \rightarrow \infty} \frac{\E[X]}{\E[X_i]} \frac{1}{T}=1-\frac{1}{\sqrt{2 \pi B}}(1+o(1)).
\]
The last equality is due to the fact that $\E[X]=B \bp{1-\frac{1}{\sqrt{2 \pi B}}\big(1+o(1)\big)}$, which result appears in multiple contexts before such as Adwords and online $B$-matching problems~\cite{devanur2012asymptotically,alaei2012online}, and correlation gap~\cite{yan2011mechanism}.
\end{proof}

%
%


\begin{proof}[Proof of Theorem~\ref{thm:supp-3}]
Lemma~\ref{lem:supp-3} suggests that when $\Del=1$, the limit on Equation~\eqref{eqn:supp-1} in Theorem~\ref{thm:supp-1} will be $1-\frac{1}{\sqrt{2 \pi B}}(1+o(1))$. Following the same procedure of applying Theorem~\ref{thm:supp-1} to prove the general results ($\Del \gg 1$) of Theorem~\ref{thm:main-3} as shown in Section~\ref{sec:large}, we can show that $\nadap$ will achieve a CR of $1-\frac{1}{\sqrt{2 \pi B}}(1+o(1))$ when $\Del=1$. The asymptotic optimality can be seen on an example similar to Example~\ref{exam:supp-2}: we just need to set up one single resource, and the single edge $e$ has $\cA_e=1$ with probability $B/T$ and $0$ otherwise. All the rest proofs are the same as shown in Lemma~\ref{lem:supp-2}.
\end{proof}

\end{document}